\documentclass{article}
\usepackage{latexsym, amscd, amsfonts, mathrsfs, amsmath, amssymb, amsthm, stmaryrd, tikz-cd, mathrsfs, bbm, url, esint, todonotes, caption, subcaption, hyperref}

\def\p{\prime}
\def\pp{{\p\p}}
\def\ppp{{\p\p\p}}
\def\bZ{\mathbb Z}

\def\bR{\mathbb R}

\def\bl{\backslash}
\def\bbl{\mathbbm 1}

\def\ol{\overline}

\def\lm{\lambda}

\def\cC{\mathcal C}

\def\sm{\sigma}

\def\al{\alpha}

\def\ot{\otimes}

\def\t{\times}

\def\Th{\Theta}

\def\de{\delta}

\def\bP{\mathbb P}

\def\cY{\mathcal Y}
\def\cX{\mathcal X}
\def\wt{\widetilde}

\def\ep{\epsilon}

\def\der{\partial}

\def\bE{\mathbb{E}}
\def\EE{\bE}
\def\cE{\mathcal{E}}

\def\cN{\mathcal{N}}

\def\Na{\nabla}

\DeclareMathOperator\Cov{\mathrm{Cov}}

\DeclareMathOperator\Var{\mathrm{Var}}

\DeclareMathOperator\Ent{\mathrm{Ent}}

\DeclareMathOperator\KL{\mathrm{KL}}
\DeclareMathOperator\SKL{\mathrm{SKL}}
\DeclareMathOperator\Unif{\mathrm{Unif}}
\DeclareMathOperator\PC{\mathrm{PC}}
\DeclareMathOperator\br{\mathrm{br}}
\DeclareMathOperator\BSC{\mathrm{BSC}}

\DeclareMathOperator\Ber{\mathrm{Ber}}
\DeclareMathOperator\Col{\mathrm{Col}}
\DeclareMathOperator\TV{\mathrm{TV}}
\DeclareMathOperator\Po{\mathrm{Po}}

\DeclareMathOperator\Capa{\mathrm{Cap}}

\newcommand{\rom}[1]{\textup{\uppercase\expandafter{\romannumeral#1}}}

\newtheorem{thm}{Theorem}
\newtheorem{lemma}[thm]{Lemma}
\newtheorem{prop}[thm]{Proposition}
\newtheorem{coro}[thm]{Corollary}

\newtheorem{claim}[thm]{Claim}

\newtheorem*{reminderInternal}{Reminder: \reminderCurrent}

\theoremstyle{definition}

\newtheorem{defn}[thm]{Definition}

\newtheorem{rmk}[thm]{Remark}

\begin{document}

\title{Non-linear Log-Sobolev inequalities for the Potts semigroup and applications to reconstruction problems}
\author{Yuzhou Gu\thanks{\texttt{yuzhougu@mit.edu}. IDSS, LIDS, and Dept. of EECS, MIT, Cambridge, MA, 02139, USA.}
\and
Yury Polyanskiy\thanks{\texttt{yp@mit.edu}. IDSS, LIDS, and Dept. of EECS, MIT, Cambridge, MA, 02139, USA.}}

\date{}



\maketitle

\begin{abstract}
	Consider the semigroup of random walk on a complete graph, which we call the Potts semigroup.
	Diaconis and Saloff-Coste~\cite{DSC96} computed the maximum of the ratio of the relative entropy and the Dirichlet form obtaining the constant $\alpha_2$ in the $2$-log-Sobolev inequality ($2$-LSI). In this paper, we obtain the best possible \textit{non-linear} inequality relating entropy
	and the Dirichlet form (i.e., $p$-NLSI, $p\ge1$). As an example, we show $\alpha_1 = 1+\frac{1+o(1)}{\log k}$.

	By integrating the $1$-NLSI we obtain the new strong data processing inequality (SDPI), which in
	turn allows us to improve results of Mossel and Peres \cite{MP03} on reconstruction thresholds for Potts models on trees. A special case is the problem of reconstructing color of the root of a $k$-colored tree given knowledge of colors of all the leaves. We show that
	to have a non-trivial reconstruction probability the branching number of the tree should be at least
	$$\frac{\log k}{\log k - \log(k-1)} = (1-o(1))k\log k.$$
	This recovers previous results (of Sly \cite{Sly09b} and Bhatnagar et al.~\cite{BVVW11}) in (slightly) more generality, but more importantly avoids the need for any coloring-specialized arguments.
	Similarly, we improve the state-of-the-art on the weak recovery threshold for the stochastic block model with $k$ balanced groups, for all $k\ge 3$.
	To further show the power of our method, we prove optimal non-reconstruction results for a broadcasting on trees model with Gaussian kernels, closing a gap left open by Eldan et al.~\cite{Eld20}.
	These improvements advocate information-theoretic methods as a useful complement to the conventional techniques originating from the statistical physics.
\end{abstract}



\newpage
\setcounter{tocdepth}{2}
\tableofcontents

\newpage

\section{Introduction}

\textbf{Log-Sobolev inequalities.}
Log-Sobolev inequalities (LSIs) are a class of inequalities bounding the rate of convergence of a Markov semigroup to its stationary distribution. They upper bound certain relative entropy (KL divergence) functions via a multiple of the Dirichlet form.

Let us introduce some standard notions.
Let $\cX$ be a finite alphabet and $K: \cX \t \cX \to [0, 1]$ be a Markov kernel, i.e., for all $x\in \cX$, we have $\sum_{y\in \cX} K(x, y) = 1$. Let $L = K-I$. We consider the semigroup $(T_t)_{t\ge 0}$, where $T_t = \exp(t L)$.
Let $\pi$ be a stationary measure for the semigroup.
For $f, g : \cX \to \bR$, the Dirichlet form is defined by
$$\cE(f, g) := -\bE_\pi[ (Lf) g] = -\sum_{x, y\in \cX} L(x, y) f(y) g(x) \pi(x).$$
For non-zero $f: \cX \to \bR_{\ge 0}$, the relative entropy is defined by
$$\Ent_\pi(f) := \bE_\pi [f \log \frac{f}{\bE_\pi[f]}] = \bE_\pi[f] D( \pi^{(f)} || \pi)$$
where $\pi^{(f)}$ is a distribution defined as $\pi^{(f)}(x) = \frac{f(x)\pi(x)}{\bE_\pi[f]}$,
and $D(P||Q)$ is the Kullback-Leibler divergence\footnote{Throughout this paper, $\log$ means natural logarithm.}
\begin{align*}
	D(P||Q):=\int \log (\frac {dP}{dQ}) dP.
\end{align*}

For $p>1$, we say the semigroup $(T_t)_{t\ge 0}$ admits $p$-log-Sobolev inequality ($p$-LSI), if for some constant $\al_p$, for all non-zero non-negative real functions $f$ on $\cX$, we have
$$\Ent_\pi(f) \le \frac 1{\al_p} \cE(f^{\frac 1p}, f^{1-\frac 1p}).$$
For $p=1$, we define $1$-LSI as
$$\Ent_\pi(f) \le \frac 1{\al_1} \cE(f, \log f).$$
The case $p=2$ is the standard log-Sobolev inequality, originally studied in Gross \cite{Gro75}.
The case $p=1$ is studied also under the name ``modified log-Sobolev inequality'' (e.g.~\cite{BT06}).

The relationship between $1$-LSI and semigroup convergence can be seen from the following identity
\begin{equation}
	\frac{d}{dt}|_{t=0} \Ent_\pi(T_t f) = -\cE(f, \log f).
\end{equation}
Therefore
\begin{equation}
	\Ent_\pi(T_t f) \le \exp(-\al_1 t) \Ent_\pi(f) \label{Eqn1LSCExp}
\end{equation}
which corresponds to a property of $T_t$ to exponentially fast relax to equillibrium (in the sense of relative entropy).

Polyanskiy and Samorodnitsky \cite{PS19} introduced non-linear $p$-log-Sobolev inequalities ($p$-NLSI), a finer description of the relationship between relative entropy and Dirichlet forms.
For $p\ge 1$, we say the semigroup satisfies $p$-LSI if for some non-negative function\footnote{\cite{PS19} requires the
function $\Phi_p$ to be concave. We do not make this assumption initially, however to extend these inequalities to
product semigroups the concavification will be necessary -- see Section~\ref{SecProd}.}. $\Phi_p: \bR_{\ge 0} \to \bR_{\ge 0}$, for all non-zero $f: \cX \to \bR_{\ge 0}$, we have
\begin{align}\label{eq:nlsi}
	\frac{\Ent_\pi(f)}{\bE_\pi[f]} \le \Phi_p(\frac{\cE(f^{\frac 1p}, f^{1-\frac 1p})}{\bE_\pi[f]}),
\end{align}
where for $p=1$, $\cE(f^{\frac 1p}, f^{1-\frac 1p})$ should be replaced with $\cE(f, \log f)$.

Non-linear $p$-log-Sobolev inequalities imply the ordinary $p$-log-Sobolev inequalities for
$$\al_p = \inf_{x>0} \frac{x}{\Phi_p(x)}.$$
When $\Phi_p$ is concave, this can be further simplified to $\al_p = (\Phi_p^\p(0))^{-1}$.

Mossel et al.~\cite{MOS13} proved that for reversible $(K, \pi)$,
\begin{align}\label{EqnpLSIComp}
	\frac{p^2(q-1)}{q^2(p-1)} \al_p \le \al_q \le \al_p
\end{align}
for $1 < q\le p\le 2$.
We discuss some general facts about dependence of $\alpha_p$ and $\Phi_p$ on $p$ in Appendix~\ref{apx:lsi_conc}.


\textbf{Potts semigroup.}
In this paper, we focus on the simplest Markov semigroup, corresponding to the random walk on a complete graph. The Markov kernel is $K(x, y) = \frac 1{k-1} \bbl\{x\ne y\}$, where $k = \#\cX$.\footnote{In the following, we always write $\cX = [k]$.}
We call it the Potts semigroup, because every operator $T_t$ in the semigroup is a ferromagnetic Potts channel. Its
stationary distribution $\pi$ is uniform on $\cX$ and its Dirichlet form is rescaled covariance:
$$ \cE(f,g) = \frac{k}{k-1} \Cov_\pi(f,g)$$

A Potts channel $\PC_\lm$ is a Markov kernel $[k] \t [k]\to [0, 1]$ defined by
\begin{align*}
	\PC_\lm(x, y) = \frac {1-\lm}k\bbl\{x\ne y\} + (\frac 1k + \frac {k-1}k\lm) \bbl\{x=y\}.
\end{align*}
We parametrize them by $\lm$ because $\lm$ is the second largest eigenvalue of $\PC_\lm$.
The valid region of $\lm$ is $[-\frac 1{k-1}, 1]$. When $\lm>0$, we call $\PC_\lm$ a ferromagnetic Potts channel; when $\lm < 0$, we call it an anti-ferromagnetic Potts channel.
One can see that $T_t$ in the Potts semigroup is exactly $\PC_{\exp(-\frac k{k-1} t)}$.

Diaconis and Saloff-Coste \cite{DSC96} computed the $2$-log-Sobolev constant
\begin{equation}\label{eq:2lsi_potts}
	\al_2 = \frac{k-2}{(k-1)\log(k-1)}\,.
\end{equation}They observed that the infimum of the ratio $\frac{\cE(f, f)}{\Ent_\pi[f^2]}$ is achieved at a two-valued function $f$, i.e., $f$ takes exactly two values. In fact, the infimum is achieved at a function $f$ where $f(1) = k-1$ and $f(i) = 1$ for $i\ne 1$.
For $p\ne 2$, it seems hard to give a closed-form expression for $\al_p$.
Goel \cite{Goe04} proved that $$\frac{k}{k-1} \le \al_1 \le (1+\frac 4{\log (k-1)})\frac{k}{k-1},$$ where the upper bound is by using a two-valued function $f$, where $f(1) = k+1$ and $f(i)=1$ for $i\ne 1$.
Bobkov and Tetali \cite{BT06} also discussed bounds on $\al_1$ and $\al_2$, proving that
$$\al_1 \ge \frac{k}{k-1} + \frac 2{\sqrt{k-1}}.$$
These computations lead to the guess that for all $p$, the best possible $p$-LSI constant $\al_p$ for the Potts semigroup is achieved at a two-valued function.
In Section \ref{SecLSI}, we prove that this is true, and in fact true for $p$-NLSIs for the Potts semigroup:
For fixed $\frac{\Ent_\pi(f)}{\bE_\pi[f]}$, the unique function (up to scalar multiplication) of the form $f(1) \ge f(2)
= \cdots = f(k)$ minimizes $\frac{\cE(f^{\frac 1p}, f^{1-\frac 1p})}{\bE_\pi[f]}$. 
As a result we get the sharpest $p$-NLSIs for the Potts semigroup for all $p\ge 1$.

We define a useful function $\psi: [0, 1] \to \bR$ as follows.
\begin{equation}\label{EqnPsi}
  \psi(x) := \log k + x\log x + (1-x) \log \frac{1-x}{k-1}.
\end{equation}
Note that $\psi(x)$ is the KL divergence between $(x, \frac{1-x}{k-1}, \ldots, \frac{1-x}{k-1})$ and $\Unif([k])$.
Simple computation shows that $\psi$ is non-negative, convex, $\psi(\frac 1k)=0$, strictly decreasing on $[0, \frac 1k]$, strictly increasing on $[\frac 1k, 1]$, and takes value in $[0, \log k]$.

We define the following useful functions.
For $p>1$, define $\xi_p: [0, 1]\to \bR$ as
\begin{align}\label{EqnXip}
	\xi_p(x) = \frac k{k-1}(1 - \frac 1k(x^{\frac 1p} + (k-1) (\frac {1-x}{k-1})^{\frac 1p})(x^{1-\frac 1p} + (k-1) (\frac {1-x}{k-1})^{1-\frac 1p})).
\end{align}
Define $\xi_1: [0, 1]\to \bR$ as
\begin{align}\label{EqnXi1}
	\xi_1(x) = \frac 1{k-1}(-\log x - (k-1)\log \frac{1-x}{k-1} + k(x\log x + (1-x) \log \frac{1-x}{k-1})).
\end{align}

For $p\ge 1$, define $b_p: [0, \log k] \to \bR$ as\footnote{In the case $k=2$, $b_p$ differs from \cite{PS19} by
a constant factor due to a different parametrization of the semigroup.}
\begin{equation}\label{eq:bp_def}
	b_p(\psi(x)) = \xi_p(x)
\end{equation}
for $x\in [\frac 1k, 1]$, where $\psi$ is defined in \eqref{EqnPsi}.

\begin{thm}[$p$-NLSI for Potts semigroup]\label{ThmNonLinpLSI}
	Fix $p\ge 1$. The Potts semigroup satisfies $p$-NLSI with $\Phi_p = b_p^{-1}$, where $b_p$ is defined in \eqref{eq:bp_def}. Furthermore, this is the best possible $p$-NLSI.

	In other words, for any $c\in [0, \log k]$, among all functions $f: [k]\to \bR_{\ge 0}$ with $\bE_\pi f = 1$, $\Ent(f) = c$, there is a unique (up to permuting the alphabet) minimizer of $\cE(f^{\frac 1p}, f^{1-\frac 1p})$ ($\cE(f, \log f)$ for $p=1$), and it is of form $(x,\frac{k-x}{k-1},\cdots,\frac{k-x}{k-1})$ with $x\in [1, k]$.

	In particular, we have
	\begin{align}
		\al_p = \inf_{x\in (\frac 1k, 1]} \frac{\xi_p(x)}{\psi(x)}.
	\end{align}

\end{thm}

As a corollary of our $1$-NLSI, we derive the second order behavior of $\al_1$ as $k$ goes to $\infty$.
\begin{prop}\label{Prop1LSCPotts}
	For $k\ge 3$, we have
	\begin{equation}
		\frac k{k-1} (1 + \frac 1{\log k}) \le \al_1 \le \frac k{k-1} (1 + \frac {1+o(1)}{\log k}). \label{Eqn1LSCPotts}
	\end{equation}
\end{prop}

\textbf{Strong data processing inequalities.}
The data processing inequality\footnote{We recall that for a pair of random variables $X,Y$ we define $I(X;Y) = D(P_{X,Y}
\| P_X P_Y)$.} states that $I(U; Y) \le I(U; X)$ for any Markov chain $U\to X \to Y$, i.e., we cannot gain information by going through a channel. It is natural to think that when the channel $X\to Y$ is noisy, we should strictly lose information, i.e., $I(U; Y) < I(U; X)$. In fact, we have $I(U; Y) \le \eta I(U; X)$, where the constant $\eta$ depends only on the channel $P_{Y|X}$, but not on the distribution of $U$ and $X$.
Such inequalities are called strong data processing inequalities (SDPIs). We distinguish between the inequalities that
depend on the distribution $P_X$ and that are independent of it. Namely, fix a stochastic matrix (conditional
distribution) $W$ and input distribution (row-vector) $Q_0$ and define
\begin{align*}\eta_{\KL}(W) = \sup_{P, Q : 0<D(P||Q) < \infty} \frac{D(PW|| QW)}{D(P||Q)}\,,\\
\eta_{\KL}(W, Q_0) = \sup_{P: 0<D(P||Q_0) < \infty} \frac{D(PW|| Q_0W)}{D(P||Q_0)}
\end{align*}
It can be shown, e.g.~\cite{PW17}, that we also have alternative characterizations:
\begin{align*}\eta_{\KL}(W) &= \sup_{U\to X\to Y} \frac{I(U;Y)}{I(U;X)}\,,\\
\eta_{\KL}(W, Q_0) &= \sup_{U\to X\to Y, X\sim Q_0} \frac{I(U;Y)}{I(U;X)}\,,
\end{align*}
where $\bP[Y=y|X=x] = W_{x,y}$. Even more generally, some channels can be shown to satisfy (for all Markov chains $U\to
X\to Y$ with $P_{Y|X}$ as before, arbitrary $U$ and fixed or arbitrary $P_X$)
\begin{equation}\label{eq:nsdpi}
	I(U;Y) \le s(I(U;X))\,,
\end{equation}
for some non-linear function $s$.
See \cite{Rag16} for more background on SDPIs and their relationship with LSIs.

From \eqref{Eqn1LSCExp} and Proposition~\ref{Prop1LSCPotts} we obtain
\begin{align}\label{eq:etakl_1lsi}
	\eta_{\KL}(\PC_\lm, \pi) \le \lm^{\frac{k-1}k \al_1} = \lm^{1+\frac{1+o(1)}{\log k}}
\end{align}
for $\lm \in [0, 1]$ and $o(1)\to0$ as $k\to \infty$. It turns out that $1$-NLSI can be seen as an infinitesimal version
of the non-linear SDPIs (see~\cite[Theorem 2]{PS19}). Thus, we can prove the best possible non-linear SDPI for Potts
channels.

\begin{thm}[Non-linear SDPI for Potts channel]\label{ThmSDPIPotts}
	Fix $\lm \in [-\frac 1{k-1}, 1]$.
	Define $s_\lm : [0, \log k] \to \bR$ as $$s_\lm(\psi(x))= \psi(\lm x + \frac{1-\lm}k),$$
	for $x\in [\frac 1k, 1]$, where $\psi$ is defined in \eqref{EqnPsi}.
	Let $\hat s_\lm$ be the concave envelope of $s_\lm$.
	For any Markov chain $U\to X\to Y$ where $X$ has uniform distribution and $X\to Y$ is the Potts channel $\PC_\lm$, we have
	$$I(U; Y) \le \hat s_\lm(I(U; X)).$$
	In particular, we have
	\begin{equation}\label{eq:coro_sdpi_potts}
		\eta_{\KL}(\PC_\lm, \pi) = \sup_{x\in (\frac 1k, 1]} \frac{\psi(\lm x + \frac{1-\lm}k)}{\psi(x)}\,.
	\end{equation}

	Furthermore, this is the best possible non-linear SDPI for Potts channels, in the sense that
	for any $c \in [0, \log k]$, there exists a Markov chain $U\to X\to Y$
	where $X$ has uniform distribution, $X\to Y$ is the Potts channel $\PC_\lm$,
	and $I(U; X)=c$, such that $I(U; Y) = \hat s_\lm(c)$.
\end{thm}


To compare the input-restricted $\eta_{\KL}$ with input-unrestricted one, in Appendix~\ref{SecPottsEtaUnres}
we compute the exact value of $\eta_{\KL}(\PC_\lm)$, and in Appendix~\ref{SecUpbd}
we prove that
\begin{equation}\label{eq:eta_res_vs_unres}
	\eta_{\KL}(\PC_\lm, \pi)< \eta_{\KL}(\PC_\lm)
\end{equation}
for $k\ge 3$ and $\lm \in [-\frac 1{k-1}, 0)\cup (0, 1)$. See Section~\ref{sec:klim} for discussion on the tightness of
the bound~\eqref{eq:etakl_1lsi}.

\begin{rmk}[Tensorization]
In Section~\ref{SecProd} we extend $p$-NLSI and SDPIs to product spaces/channels.
In these results, the functions $\check b_p$ (convexification of $b_p$) and $\hat s_\lm$
(concavification of $s_\lm$) appear naturally.
When $k=2$, $b_p$ is already convex, and $s_\lm$ is concave, leading to many good properties
for the hypercube and for binary symmetric channels (e.g.~Mrs.~Gerber's Lemma \cite{WZ73}).
However, as shown in Proposition \ref{PropNonConv}, these properties do not hold anymore for $k\ge 3$, implying a
different structure of extremal distributions that are the slowest to relax to equillibrium as $t\to\infty$ in
$T_t^\otimes n$, see Section\ref{sec:linpiece}.
\end{rmk}

\textbf{Applications.}

One of the implications of NLSIs are improved hypercontractivity inequalities for functions in $[k]^n$ supported on
subsets of cardinality $k^{(1-\epsilon)n}$ -- this was established generally (for any semigroup) in~\cite{PS19}. Here,
we show how NLSIs can be used to close the gap (between functional-analytic proofs and explicit combinatorics
of~\cite{Lin64}) in the edge-isoperimetric inequality for the $[k]^n$ -- see Section~\ref{sec:edgeisop}.

Similarly, SDPIs have numerous applications. Originally introduced to study certain multi-user data-compression questions in
information theory, they have been since adopted in many different scenarios. For example, Evans and Schulman
\cite{ES99} use SDPIs to investigate fundamental limits of fault tolerant computing.  Polyanskiy and Wu \cite{PW17, PW18} further developed the idea and related the amount of information transmitted in a directed or undirected
graphical model in terms of the percolation probability (existence of an open path) on the same network. Other
notable applications include distributed estimation~\cite{XR15,BGM16} and communication complexity~\cite{HLPS19}.

More directly related to our paper is the work of~\cite{EKPS00} which applies an SDPI (for a Potts channel with $k=2$)
to bound the threshold on for reconstruction problem for the Ising model on a tree. We describe a more general class of
such questions.

Consider a tree with a marked root.
Each vertex of the tree has a random spin in $[k]$, generated in the following way: the root spin is generated according to some known distribution, and for each vertex, its spin is generated from its parent's spin, through some channel $M$.
We say the problem has reconstruction if given spins of all nodes far away enough from the root, we can guess the root spin better than guessing from the initial distribution.
Equivalently, the problem has reconstruction if the mutual information between the root spin and the spins of all nodes distance $d$ away from the root goes to a non-zero limit, as $d$ goes to $\infty$.
The problem has non-reconstruction otherwise.

Reconstruction problems trees have been studied for a long time. Kesten and Stigum \cite{KS66} proved the so-called Kesten-Stigum bound, a reconstruction result based on the second eigenvalue of the channel.
Physicists study the problem from a spin glass theoretic perspective (e.g., M{\'e}zard and Montanari \cite{MM06}).
Based on careful analysis of the evolution of magnetization, several authors have obtained very tight reconstruction thresholds for various models (e.g., Sly \cite{Sly09a, Sly09b}, Bhatnagar et al.~\cite{BST10}, Liu and Ning \cite{LN19}).

In this paper, we attack this problem using SDPIs. The method of~\cite{EKPS00} showed that reconstruction is impossible
if $\br(T) \eta_{\KL}(M) < 1$. Here we improve their result by considering the input-restricted $\eta_{\KL}$. We note that the
idea of using input-restricted (but not for $\eta_{\KL}$) has appeared in~\cite{FK09a} -- see Remark~\ref{RmkColorComparison}
below.
\begin{thm}\label{ThmTreeGeneral}
	Consider the broadcast model on a tree $T$ with channel $M$.
	Let $q^*$ be a stationary distribution, i.e., $q^* M = q^*$.
	Let $M^\vee$ denote the reverse channel, i.e., any channel that satisfies $q^*_j M^\vee_{j,i} = q^*_i M_{i,j}$ for all $i,j\in [k]$.
	Then the model has non-reconstruction if
	$$\eta_{\KL}(M^\vee, q^*) \br(T) < 1,$$
	where $\br(T)$ is the branching number of $T$, whose definition is given in Definition \ref{DefnBranch}.
\end{thm}
Our method is very simple, is non-asymptotic, works for the branching number, and is often very tight.
Previous results often impose additional restrictions on the tree (e.g., regular tree, or Galton-Watson with Poisson offspring distribution), and some only work as expected degree goes to infinity. We discuss in more detail in Section \ref{SecTreeGeneral}.

Applying Theorem \ref{ThmTreeGeneral} to the Potts model, we achieve improved non-reconstruction results for Potts models on a tree.
\begin{thm}\label{ThmPotts}
	Consider the Potts model on a tree $T$.
	We have non-reconstruction whenever
	$$\eta_{\KL}(\PC_\lm, \pi) \br(T) < 1,$$
	where
	$\eta_{\KL}(\PC_\lm, \pi)$ is given by~\eqref{eq:coro_sdpi_potts}.
\end{thm}
Theorem \ref{ThmPotts} strictly improves over the explicit bound of Mossel and Peres \cite{MP03}. In the special case
where the channel is the coloring channel ($\lm = -\frac 1{k-1}$), Theorem \ref{ThmPotts} recovers the reconstruction
threshold up to the first order, which was previously obtained by Sly \cite{Sly09b} and Bhatnagar et al.~\cite{BVVW11}
using more complicated methods. (We note, however, that a major focus of those works was to obtain the lower-order
terms.)
More detailed analysis is done in Section \ref{SecPotts}.

Last but not least, we consider the problem of weak recovery for the stochastic block model with $k$ communities ($k$-SBM).
In $k$-SBM, $n$ vertices each receive a random spin in $[k]$, and then a random graph is constructed, such that
(1) for two vertices with the same spin, there exists an edge with probability $\frac an$;
(2) for two vertices with different spins, there exists an edge with probability $\frac bn$.
The model is said to have weak recovery, if given the random graph, we can partition the vertices into $k$ parts, such that the partition is correct (up to relabeling the parts) for at least $\ep n$ vertices, for some absolute constant $\ep$.

For $k=2$, the weak recovery threshold is known: If $(a-b)^2 > 2(a+b)$, weak recovery is possible (Massouli{\'e} \cite{Mas14}, Mossel et al.~\cite{MNS18}); if $(a-b)^2 < 2(a+b)$, weak recovery is impossible (Mossel et al.~\cite{MNS15}).
For $k\ge 3$, the weak recovery threshold is not completely determined.
The stochastic block model is called assortative if $a>b$, and disassortative if $a<b$. The disassortative case was
solved by Coja-Oghlan et al.~\cite{CKPZ18}.
For the assortative case, the current best impossibility result for general $k$ is by Banks et al.~\cite{BMNN16}, which says weak recovery does not hold whenever
\begin{equation}\label{EqnSBMBMNN}
	\frac{(a-b)^2}{a+(k-1)b} < \frac{2k\log (k-1)}{k-1}.
\end{equation}

By a standard reduction from the Potts model, we achieve improved impossibility results for weak recovery for the $k$-stochastic models.
\begin{thm}\label{ThmSBM}
	Weak recovery of the stochastic block model is impossible if
	$$d \eta_{\KL}(\PC_\lm, \pi) < 1,$$
	where $d = \frac{a + (k-1)b}{k}$, $\lm = \frac{a-b}{a+(k-1)b}$, and $\eta_{\KL}$ is given
	by~\eqref{eq:coro_sdpi_potts}.
\end{thm}
We discuss in more detail, and show that this improves over previous results, in Section \ref{SecSBM}.

\textbf{Organization.}
In Section \ref{SecLSI} we prove the sharpest $p$-NLSIs for the Potts semigroup (Theorem \ref{ThmNonLinpLSI}), and compute the input-restricted KL divergence contraction coefficients of all Potts channels (Theorem \ref{ThmSDPIPotts}).

In Section \ref{SecProd} we discuss tensorization of $p$-NLSIs for the Potts semigroup, and non-linear SDPI for Potts channels.

In Section \ref{SecTreeGeneral} we prove a non-reconstruction result for a general class of broadcast models on trees, based on strong data processing inequalities (Theorem \ref{ThmTreeGeneral}).

In Section \ref{SecPotts} we apply Theorem \ref{ThmTreeGeneral} to the Potts model on a tree (Theorem \ref{ThmPotts}). We show that this improves previous non-reconstruction results. For a special case, the random coloring model on a tree, we obtain non-reconstruction results for arbitrary trees, generalizing previous results which work only for restricted classes of trees.

In Section \ref{SecSBM}, by a standard reduction from the Potts model, we prove impossibility results for weak recovery of the stochastic block model (Theorem \ref{ThmSBM}). This results in improvements for the best known bounds for the $k$-SBM.


\section{Non-linear $p$-log-Sobolev inequalities for the Potts semigroup} \label{SecLSI}
In this section, we prove $p$-NLSIs for the Potts semigroup for $p\ge 1$.
Because the form of the $p$-LSIs are slightly different for $p\ne 1$ and $p=1$, we prove them separately.

Recall our setting. Alphabet $\cX = [k]$ for some positive integer $k\ge 2$.
The Potts semigroup $T_t = \exp(L t)$ for generator $L = \frac 1{k-1} \bbl\{x\ne y\} - \bbl\{x=y\}$.
The stationary distribution is $\pi = \Unif([k])$.
The Dirichlet form is
$$\cE(f, g) = -\bE_\pi[(L f) g] = - \frac 1{k(k-1)} (\sum_x f(x)) (\sum_y g(y)) + \frac 1{k-1}\sum_x f(x) g(x).$$
Relative entropy is
$$\Ent_\pi(f) = \bE_\pi [f \log \frac{f}{\bE_\pi [f]}].$$
The non-linear $p$-log-Sobolev inequality says
$$\frac{\Ent_\pi(f)}{\bE_\pi[f]} \le \Phi_p(\frac{\cE(f^{\frac 1p}, f^{1-\frac 1p})}{\bE_\pi[f]})$$
for some concave $\Phi_p$, where for $p=1$, RHS is replaced with $\Phi_p(\frac{\cE(f, \log f)}{\bE_\pi[f]})$.
Because both sides of the inequality are fixed under scalar multiplication, we can wlog restrict $f$ to be a distribution $P$.
Then the relative entropy is $$\Ent_\pi(P) = \frac 1k D (P ||\pi) = \frac 1k(\log k - H(P)).$$

\subsection{Non-linear $p$-log-Sobolev inequality for $p>1$}
We prove Theorem \ref{ThmNonLinpLSI} for $p>1$. Before proving the theorem we show the following.

\begin{prop}\label{PropNonLinpLSI}
	Fix $r\in (0, 1)$ and $c\in [0, \log k]$.
	Among all distributions $P = (p_1, \ldots, p_k)$ with $H(P) = c$, the distribution of form $P = (x, \frac {1-x}{k-1}, \ldots, \frac{1-x}{k-1})$ with $x\in [\frac 1k, 1]$ achieves maximum $\sum_i p_i^r$.
	Furthermore, up to permutation of the alphabet this is the unique maximum-achieving distribution.
\end{prop}
\begin{proof}
	The result for $c\in \{0, \log k\}$ is obvious. In the following, assume that $c\in (0, \log k)$.
	Write $F(P) := \sum_i p_i^r$.
	The set $\{P: H(P)=c\}$ is compact, so the maximum value of $F(P)$ is achieved at some point $P = (p_1, \ldots, p_k)$.

	We prove in several steps. In Step 0, we prove that if $p_i=0$ for some $i$, then there can be at most two different values of $p_i$'s. In Step 1, we prove that if $p_i>0$ for all $i$, then there can be at most two different values of $p_i$'s. In Step 2, we prove that one of the two different values must have multiplicity one, thus finishing the proof of the proposition.

	\textbf{Step 0.}
	\begin{claim}\label{ClaimpLSIZero}
		Fix $a,b>0$ and $r\in (0, 1)$. Among all solutions $u, v, w\in [0, 1]$ with $u+v+w=a$ and $-u \log u - v \log v - w \log w=b$, the maximum of $u^r + v^r + w^r$ is not achieved at a point where $0=u<v<w$.
	\end{claim}
	\begin{proof}
		Suppose the maximum is achieved at such a point $(u_0, v_0, w_0)$ where $0=u_0<v_0<w_0$.
		Extend it to a curve $(u, v=v(u), w=w(u))$ on $u\in [0, \ep)$ for some $\ep>0$, such that $u<v<w$ for all $u$, satisfying
		\begin{align}
			u + v + w &= a, \label{EqnClaimpLSIZeroA}\\
			-u\log u - v \log v - w \log w &= b \label{EqnClaimpLSIZeroB},
		\end{align}
		and
		$$v(0) = v_0, w(0) = w_0.$$
		We prove that $$f(u) := u^r + v^r + w^r$$
		decreases as $u$ approaches $0^+$, for small enough $u$.

		By taking derivative of \eqref{EqnClaimpLSIZeroA} and \eqref{EqnClaimpLSIZeroB}, one can compute that
		\begin{align*}
			v^\p(u) &= \frac{\log w - \log u}{\log v - \log w},\\
			w^\p(u) &= \frac{\log u - \log v}{\log v - \log w}.
		\end{align*}
		Therefore
		\begin{align*}
			f^\p(u) &= r (u^{r-1} + v^{r-1} v^\p(u) + w^{r-1} w^\p(u)) \\
			& = r (u^{r-1} + v^{r-1} \frac{\log w-\log u}{\log v-\log w} + w^{r-1} \frac{\log u-\log v}{\log v-\log w}).
		\end{align*}
		Because $0<v_0<w_0$, the term $u^{r-1}$ dominates the sum, and $f^\p(u) > 0$ for small enough $u>0$.
		Therefore the maximum of $f$ is not achieved at $u=0$.
	\end{proof}
	By Claim \ref{ClaimpLSIZero}, if $p_i=0$ for some $i$, then there can be at most two different values of $p_i$'s.

	\textbf{Step 1.}
	\begin{claim} \label{ClaimpLSIDet}
		If $u, v, w \in (0, 1)$ are all different, then
		\begin{align*}
			\det \begin{pmatrix}
				1 & \log u & u^{r-1} \\
				1 & \log v & v^{r-1} \\
				1 & \log w & w^{r-1}
			\end{pmatrix}\ne 0.
		\end{align*}
	\end{claim}
	\begin{proof}[Proof of Claim]
		Suppose $\det = 0$. Then for some $a, b\in \bR$, the equation
		$x^{r-1} + a \log x = b$ has at least three distinct solutions $x\in (0, 1)$.
		However $$\frac{\der}{\der x} (x^{r-1} + a\log x) = (r-1) x^{r-2} + \frac ax$$
		is smooth on $(0, 1)$, and takes zero at most once.
		So $x^{r-1} + a\log x$ takes each value at most once on $(0, 1)$. Contradiction.
	\end{proof}
	By Lagrange multipliers, the three vectors
	\begin{align*}
		\Na F(P) &= (r p_i^{r-1})_{i\in [k]},\\
		\Na H(P) &= (-1-\log p_i)_{i\in [k]},\\
		\Na \sum_{i\in [k]} p_i &= \bbl
	\end{align*}
	should be linear dependent.
	By Step 0 and Claim \ref{ClaimpLSIDet}, there can be at most two different values of $p_i$'s.

	So we can assume that $p_1 = \cdots = p_m = x$, $p_{m+1} = \cdots = p_k = \frac {1-mx}{k-m}$ for some $m\in [k-1]$, $x\in (\frac 1k, \frac 1m]$.

	\textbf{Step 2.}
	For $P$ of the above form, we have
	\begin{align*}
		-H(P) &= mx\log x + (1-mx) \log \frac{1-mx}{k-m},\\
		F(P) &= m x^r + (k-m) (\frac{1-mx}{k-m})^r.
	\end{align*}
	We smoothly continue both functions so that $m$ can take any real value in $[1, k-1]$.
	\begin{claim}\label{ClaimpLSIDer}
		For $m\in (1, k-1]$ and $x\in (\frac 1k, \frac 1m)$, we have
		$$-\frac{\der}{\der x} H(P) > 0$$
		and
		$$\frac{\der}{\der x} H(P) \frac{\der}{\der m} F(P) - \frac{\der}{\der m} H(P) \frac{\der}{\der x} F(P) > 0.$$
	\end{claim}
	\begin{proof}
		We have
		\begin{align*}
			-\frac{\der}{\der x}H(P) &= m(\log x - \log \frac{1-mx}{k-m}) > 0,\\
			-\frac{\der}{\der m}H(P) &= \frac{1-kx}{k-m} + x(\log x -\log \frac{1-mx}{k-m}), \\
			\frac{\der}{\der x} F(P) &= rm(x^{r-1} - (\frac{1-mx}{k-m})^{r-1}), \\
			\frac{\der}{\der m} F(P) &= x^r + (r \frac{1-kx}{1-mx} -1) (\frac{1-mx}{k-m})^r.
		\end{align*}
		Let $a = \frac{kx-1}{1-mx}$.
		Then
		\begin{align*}
			G(P) &:= \frac{\der}{\der x} H(P) \frac{\der}{\der m} F(P) - \frac{\der}{\der m} H(P) \frac{\der}{\der x} F(P)\\
			& = (r-1)m (x^r - (\frac{1-mx}{k-m})^r)(\log x - \log \frac{1-mx}{k-m}) \\
			&- rm \frac{1-kx}{k-m}(x^{r-1} - (\frac{1-mx}{k-m})^{r-1}) \\
			& = x^{-r} m ((r-1) (1-(a+1)^{-r}) \log(a+1) - r \frac{a}{a+1} (1-(a+1)^{1-r})).
		\end{align*}
		The result then follows from Claim \ref{ClaimElem}.
	\end{proof}
	\begin{claim}\label{ClaimElem}
		For all $r\in (0, 1)$ and $a>0$ we have $$(r-1) (1-(a+1)^{-r}) \log(a+1) - r \frac{a}{a+1} (1-(a+1)^{1-r}) > 0.$$
	\end{claim}
	\begin{proof}
		Let $$f(a) := (r-1) (1-(a+1)^{-r}) \log(a+1) - r \frac{a}{a+1} (1-(a+1)^{1-r}).$$
		Because $\lim_{a\to 0^+} f(a) = 0$, it suffices to prove that $f^\p(a) > 0$.
		\begin{align*}
			f^\p(a) &= (a+1)^{-r-1} (1-r - (a(1-r)+1)((a+1)^{r-1} - r)- (1-r)r \log (a+1))\\
			&=: (a+1)^{-r-1} g(a).
		\end{align*}
		Because $\lim_{a\to 0^+} g(a) = 0$, it suffices to prove that $g^\p(a) > 0$.
		\begin{align*}
			g^\p(a) = \frac{a r (1-r) (1-(a+1)^{r-1})}{a+1} > 0.
		\end{align*}
	\end{proof}
	Now let us return to the proof of Proposition \ref{PropNonLinpLSI}.
	The set of $(m, x)$ where $m\in [1, k-1]$, $x\in (\frac 1k, \frac 1m]$, and $H(P) = c$ can be parametrized as a curve
	$(m, x=x(m))$ for $m\in [1, m_c]$ for some constant $m_c$.
	Along the curve, $F(P)$ is continuous, and by Claim \ref{ClaimpLSIDer}, is decreasing in $m$.
	Therefore $F(P)$ is maximized at $m=1$. This finishes the proof.
\end{proof}

\begin{proof}[Proof of Theorem \ref{ThmNonLinpLSI} for $p>1$]
	For a distribution $P = (p_1, \ldots, p_k)$, we have
	\begin{align*}
		\cE(P^{\frac 1p}, P^{1-\frac 1p}) = \frac 1{k-1} (1 - \frac 1k (\sum_i p_i^{\frac 1p}) (\sum_i p_i^{1-\frac 1p})).
	\end{align*}
	By Proposition \ref{PropNonLinpLSI}, for fixed value of $\Ent_\pi(P)$, the unique distribution of the form $(x, \frac{1-x}{k-1}, \ldots, \frac{1-x}{k-1})$ with $x\in [\frac 1k, 1]$ minimizes $\cE(P^{\frac 1p}, P^{1-\frac 1p})$.
	Therefore for any non-zero non-negative $f$, we have $$b_p(\frac{\Ent_\pi(f)}{\bE_\pi[f]}) \le \frac{\cE(f^{\frac 1p}, f^{1-\frac 1p})}{\bE_\pi[f]}.$$
	So $p$-NLSI holds with $\Phi_p = b_p^{-1}$.
	The statement about optimality is immediate from the above discussions.
\end{proof}

\subsection{Non-linear $1$-log-Sobolev inequality}
We prove Theorem \ref{ThmNonLinpLSI} for $p=1$.
Before proving the theorem we show the following.
\begin{prop}\label{PropNonLin1LSI}
	Fix $0\le c \le \log k$.
	Among all distributions $P$ with $H(P) = c$, the distribution of form $P = (x, \frac {1-x}{k-1}, \ldots, \frac{1-x}{k-1})$ with $x\in [\frac 1k, 1]$ achieves maximum $\sum_i \log p_i$.
	Furthermore, up to permutation of the alphabet this is the unique minimum-achieving distribution.
\end{prop}
\begin{proof}
	The result for $c\in \{0, \log k\}$ is obvious. In the following, assume that $0 < c < \log k$.
	Write $F(P) := \sum_i \log p_i$.
	The set $\{P: H(P) = c\}$ is compact, so the maximum value of $F(P)$ is achieved at some point $P = (p_1, \ldots, p_k)$.

	We prove in several steps. In Step 0, we prove that $p_i>0$ for all $i$. In Step 1, we prove that there can be at most two different values of $p_i$'s. In Step 2, we prove that one of the two different values must have multiplicity one, thus finishing the proof of the proposition.

	\textbf{Step 0.}
	If $p_i=0$ for some $i$, then $F(P) = -\infty$. So $\min_{i\in [k]} p_i > 0$.

	\textbf{Step 1.}
	\begin{claim}\label{Claim1LSIDet}
		If $u, v, w \in (0, 1)$ are all different, then
		$$\det \begin{pmatrix} 1 & \log u & \frac 1u\\ 1 & \log v & \frac 1v\\ 1 & \log w & \frac 1w \end{pmatrix}\ne 0.$$
	\end{claim}
	\begin{proof}[Proof of Claim]
		Suppose $\det = 0$. Then for some $a, b\in \bR$, the equation $\frac 1x + a \log x = b$ has at least three distinct solutions $x\in (0, 1)$.
		However, $\frac{\der}{\der x} (\frac 1x + a \log x) = - \frac 1{x^2} + \frac ax$ is smooth on $(0, 1)$, and takes zero at most once.
		So $\frac 1x + a \log x$ takes each value at most once on $(0, 1)$. Contradiction.
	\end{proof}
	By Lagrange multipliers, the three vectors
	\begin{align*}
		&\Na F(P) = (\frac 1{p_i})_{i\in [k]},\\
		&\Na H(P) = (-1-\log p_i)_{i\in [k]},\\
		&\Na \sum_{i\in [k]} p_i = \bbl
	\end{align*}
	should be linear dependent. By Claim \ref{Claim1LSIDet}, there can be at most two different values of $p_i$'s.

	So we can assume that $p_1 = \cdots = p_m = x$, $p_{m+1} = \cdots = p_k = \frac {1-mx}{k-m}$ for some $m\in [k-1]$, $x\in (\frac 1k, \frac 1m)$.

	\textbf{Step 2.}
	For $P$ of the above form, we have
	\begin{align*}
		-H(P) &= m x \log x + (1-m x) \log \frac{1-m x}{k-m},\\
		F(P) &= m \log x + (k-m) \log \frac {1-mx} {k-m}.
	\end{align*}
	We smoothly continue both functions so that $m$ can take any real value in $[1, k-1]$.
	\begin{claim}\label{Claim1LSIDer}
		For $m \in (1, k-1]$ and $x\in (\frac 1k, \frac 1m)$, we have
		$$- \frac{\der}{\der x} H(P) > 0$$
		and
		$$\frac{\der}{\der x} H(P) \frac{\der}{\der m} F(P) - \frac{\der}{\der m} H(P) \frac{\der}{\der x} F(P) > 0.$$
	\end{claim}
	\begin{proof}[Proof of Claim]
		Let $f(x) = \log x - \log \frac{1-mx}{k-m}$.
		Then we have
		\begin{align*}
			-\frac{\der}{\der x} H(P) &= m f(x) > 0,\\
			-\frac{\der}{\der m} H(P) &= \frac {1-kx}{k-m} + x f(x),\\
			\frac{\der}{\der x} F(P) &= \frac{m(1-kx)}{x(1-mx)}\\
			\frac{\der}{\der m} F(P) &= \frac{1-kx}{1-mx} + f(x).
		\end{align*}
		So
		\begin{align*}
			G(P) &:= \frac{\der}{\der x} H(P) \frac{\der}{\der m} F(P) - \frac{\der}{\der m} H(P) \frac{\der}{\der x} F(P)\\
			& = m (\frac{(1-kx)^2}{x(k-m)(1-mx)} - f(x)^2).
		\end{align*}
		Let $a = \frac {kx-1}{1-mx}$.
		Then $G(P) = m (\frac{a^2}{1+a} - \log^2 (a+1))$.
		Because $a > 0$, we have $G(P) > 0$ by Lemma \ref{LemmaElem}.
	\end{proof}
	The set of $(m, x)$ where $m\in [1, k-1]$, $x\in (\frac 1k, \frac 1m]$, and $H(P) = c$ can be parametrized as a curve
	$(m, x=x(m))$ for $m\in [1, m_c]$ for some constant $m_c$.
	Along the curve, $F(P)$ is continuous, and by Claim \ref{Claim1LSIDer}, is decreasing in $m$.
	Therefore $F(P)$ is maximized at $m=1$. This finishes the proof.
\end{proof}
\begin{lemma}\label{LemmaElem}
	For $a\in \bR_{>-1}$, we have $\frac{a^2}{1+a} \ge \log^2 (a+1)$.
	Equality holds only when $a=0$.
\end{lemma}
\begin{proof}
	We start from the well-known fact that $a \ge \log (a+1)$ for $a\in \bR_{>-1}$ (and equality holds only when $a=0$).
	Let $f(a) = a(a+2) - (2a+2) \log (a+1)$.
	We have $f(0)=0$ and $f^\p(a) = 2(a-\log (a+1)) \ge 0$ for $a\in \bR_{>-1}$ (and equality holds only when $a=0$).
	So $f$ is negative on $(-1, 1)$ and positive on $(1, \infty)$.

	Let $g(a) = \frac{a^2}{1+a} - \log^2 (a+1)$.
	Clearly $g(0) = 0$.
	Because $g^\p(a) = \frac{f(a)} {(a+1)^2}$, $g$ is decreasing on $(-1, 1]$ and increasing on $[1, \infty)$.
	So $g(a) \ge 0$ for all $a\in \bR_{>-1}$, and equality holds only when $a=0$.
\end{proof}

\begin{proof}[Proof of Theorem \ref{ThmNonLinpLSI} for $p=1$]
	For a distribution $P = (p_1, \ldots, p_k)$, we have
	\begin{align*}
		\cE(P, \log P) = \frac 1{k-1} \sum_{i\in [k]} p_i \log p_i  - \frac 1{k(k-1)} \sum_{i\in [k]} \log p_i.
	\end{align*}
	By Proposition \ref{PropNonLin1LSI}, for fixed value of $\Ent_\pi(P)$, the unique distribution of the form $(x, \frac{1-x}{k-1}, \ldots, \frac{1-x}{k-1})$ with $x\in [\frac 1k, 1]$ minimizes $\cE(P, \log P)$.
	Therefore for any non-zero non-negative $f$, we have
	\begin{align*}
		b_1(\frac{\Ent_\pi(f)}{\bE_\pi[f]}) \le \frac{\cE(f, \log f)}{\bE_\pi[f]}.
	\end{align*}
	So $1$-NLSI holds for $\Phi_1 = b_1^{-1}$.
	The statement about optimality is immediate from the above discussions.
\end{proof}


\subsection{Input-restricted non-linear SDPI for Potts channels}
In this section, we prove Theorem \ref{ThmSDPIPotts}.

The subset of Potts channels corresponding to $\lambda \ge 0$ (i.e. ferromagnetic Potts channels) form a semigroup. For
the semigroups, the optimal $1$-NLSI is an ``infinitesimal version'' of the input-restricted non-linear SDPI.
Consequently, by integrating the former we can get the latter (this is
formalized in the first part of the proof below). Surprisingly, the result also extends beyond the
semigroup to all of the Potts channels, namely we have the following.

\begin{prop}\label{PropSDPIPotts}
	Let $\lm \in [-\frac 1{k-1}, 1]$. Fix $0 \le c \le \log k$.
	Among all distributions $P$ with $H(P) = c$, the distribution of form $P = (x, \frac{1-x}{k-1}, \ldots, \frac{1-x}{k-1})$ with $x\in [\frac 1k, 1]$ achieves minimum $H(\PC_\lm \circ P)$.
	Furthermore, when $\lm\not \in \{0, 1\}$, up to permutation of the alphabet this is the unique minimum-achieving distribution.
\end{prop}
\begin{proof}
	The result for $\lm \in \{0, 1\}$ is obvious. In the following assume that $\lm\not \in \{0, 1\}$.
	The result for $c\in \{0, \log k\}$ is obvious. In the following assume that $c\not \in \{0, \log k\}$.
	The set $\{P : H(P) = c\}$ is compact, so the minimum value of $H(\PC_\lm \circ P)$ is achieved at some point $P = (p_1, \ldots, p_k)$.

	We prove in several steps. In Step 0, we prove that if $p_i=0$ for some $i$, then there can be at most two different values of $p_i$'s. In Step 1, we prove that if $p_i>0$ for all $i$, then there can be at most two different values of $p_i$'s. In Step 2, we prove that one of the two different values must have multiplicity one, thus finishing the proof of the proposition.

	\textbf{Step 0.}
	\begin{claim}\label{ClaimSDPIZero}
		Fix $a, b, d>0$ and $c\in \bR_{>-d} \bl \{0\}$.
		Among all solutions $u, v, w\in [0, 1]$ with $u+v+w=a$ and $-u\log u - v \log v - w\log w = b$, the maximum of $$(c u+d) \log (c u+d)+(c v+d) \log (c v+d)+(c w+d) \log (c w+d)$$
		is not achieved at a point where $0=u<v<w$.
	\end{claim}
	\begin{proof}
		Suppose the maximum is acheived at such a point $(u_0, v_0, w_0)$ where $0=u_0<v_0<w_0$.
		Extend it to a curve $(u, v=v(u), w = w(u))$ on $u\in [0, \ep)$ for some $\ep>0$, such that $u<v<w$ for all $u$, satisfying
		\begin{align}
			u + v + w &= a, \label{EqnClaimSDPIZeroA}\\
			-u\log u - v \log v - w \log w &= b \label{EqnClaimSDPIZeroB},
		\end{align}
		and $v(0)=v_0$, $w(0)=w_0$.

		We prove that $$f(u) := (c u+d) \log (c u+d)+(c v+d) \log (c v+d)+(c w+d) \log (c w+d)$$ decreases as $u$ approaches $0^+$ for small enough $u$.

		By taking derivative of \eqref{EqnClaimSDPIZeroA} and \eqref{EqnClaimSDPIZeroB}, one can compute that
		\begin{align*}
			v^\p(u) &= \frac{\log w - \log u}{\log v - \log w},\\
			w^\p(u) &= \frac{\log u - \log v}{\log v - \log w}.
		\end{align*}
		Therefore
		\begin{align*}
			f^\p(u) &= c(\log (cu+d) + \log (cv+d) v^\p(u) + \log (cw+d) w^\p(u)) \\
			& = c(\log (cu+d) + \log(cv+d) \frac{\log w - \log u}{\log v - \log w} + \log(cw+d) \frac{\log u - \log v}{\log v - \log w}).
		\end{align*}
		Because $0<v_0<w_0$, terms involving $\log u$ dominates the sum. The dominating term is
		\begin{align*}
			-c \log u\frac{\log (cv+d) - \log(cw+d)}{\log v - \log w} > 0.
		\end{align*}
		Therefore the maximum of $f$ is not achieved at $u=0$.
	\end{proof}
	By Claim \ref{ClaimSDPIZero}, if $p_i=0$ for some $i$, then there can be at most two different values of $p_i$'s.

	\textbf{Step 1.}
	\begin{claim}\label{ClaimSDPIDet}
		If $u, v, w \in (0, 1)$ are all different, then
		$$\det \begin{pmatrix} 1 & \log u & \log (\lm u + \frac{1-\lm} k)\\ 1 & \log v & \log (\lm v + \frac{1-\lm} k)\\ 1 & \log w & \log (\lm w + \frac{1-\lm} k) \end{pmatrix}\ne 0.$$
	\end{claim}
	\begin{proof}[Proof of Claim]
		Suppose $\det = 0$. Then for some $a, b\in \bR$, the equation $$\log(\lm x + \frac{1-\lm}k) + a \log x = b$$ has at least three distinct solutions $x\in (0, 1)$.
		However, $$\frac{\der}{\der x} (\log(\lm x + \frac{1-\lm}k) + a \log x) = \frac{\lm}{\lm x + \frac{1-\lm}k} + \frac ax$$
		is smooth on $(0, 1)$, and takes zero at most once.
		So $$\log(\lm x + \frac{1-\lm}k) + a \log x$$ takes each value at most twice on $(0, 1)$. Contradiction.
	\end{proof}
	By Lagrange multipliers, the three vectors
	\begin{align*}
		&\Na H(\PC_\lm \circ P) = (-\lm \log (\lm p_i + \frac {1-\lm}k) - \lm)_{i\in [k]},\\
		&\Na H(P) = (-1-\log p_i)_{i\in [k]},\\
		&\Na \sum_{i\in [k]} p_i = \bbl
	\end{align*} should be linear dependent. By Claim \ref{ClaimSDPIDet}, there can be at most two different values of $p_i$'s.

	So we can assume that $p_1 = \cdots = p_m = x$, $p_{m+1} = \cdots = p_k = \frac{1-mx}{k-m}$ for some $m\in [k-1]$, $x\in (\frac 1k, \frac 1m]$.

	\textbf{Step 2.}
	For $P$ of the above form, we have
	\begin{align*}
		-H(P) &= m x \log x + (1-m x) \log \frac{1-mx}{k-m},\\
		-H(\PC_\lm \circ P) &= m (\lm x + \frac{1-\lm}k) \log (\lm x + \frac{1-\lm}k) \\
		&+ (k-m) (\lm \frac{1-mx}{k-m} + \frac{1-\lm}k) \log (\lm \frac{1-mx}{k-m} + \frac{1-\lm}k).
	\end{align*}
	We smoothly continue both functions so that $m$ can take any real value in $[1, k-1]$.
	\begin{claim}\label{ClaimSDPIDer}
		For $m \in (1, k-1]$ and $x\in (\frac 1k, \frac 1m)$, we have
		$$- \frac{\der}{\der x} H(P) > 0$$
		and
		$$\frac{\der}{\der m} H(P) \frac{\der}{\der x} H(\PC_\lm \circ P) -  \frac{\der}{\der x} H(P) \frac{\der}{\der m} H(\PC_\lm \circ P) > 0.$$
	\end{claim}
	\begin{proof}[Proof of Claim]
		Let $$f(x) = \log x - \log \frac{1-mx}{k-m}.$$
		Then
		\begin{align*}
			-\frac{\der}{\der x} H(P) &= m f(x) > 0,\\
			-\frac{\der}{\der m} H(P) &= \frac{1-kx}{k-m} + x f(x),\\
			-\frac{\der}{\der x} H(\PC_\lm \circ P) &= \lm m f(\lm x + \frac{1-\lm}k),\\
			-\frac{\der}{\der m} H(\PC_\lm \circ P) &= \lm \frac{1-kx}{k-m} +(\lm x + \frac{1-\lm}k) f(\lm x + \frac{1-\lm}k),
		\end{align*}
		and
		\begin{align*}
			G(P) := & \frac{\der}{\der m} H(P) \frac{\der}{\der x} H(\PC_\lm \circ P) -  \frac{\der}{\der x} H(P) \frac{\der}{\der m} H(\PC_\lm \circ P)\\
			& = \lm m \frac{1-kx}{k-m}(f(\lm x +\frac{1-\lm}k) - f(x)) - m f(x) f(\lm x +\frac{1-\lm}k) \frac{1-\lm}k.
		\end{align*}

		\begin{align*}
			\frac{\der}{\der \lm} \frac{G(P)}{m \lm f(x) f(\lm x + \frac{1-\lm}k)} &= \frac{1}{k\lm^2} + \frac{1-kx}{k-m} \frac{\frac{\der}{\der \lm} f(\lm x + \frac{1-\lm}k)}{f(\lm x + \frac{1-\lm}k)^2}.
		\end{align*}
		Note that
		\begin{enumerate}
			\item $$\frac{G(P)}{m \lm f(x) f(\lm x + \frac{1-\lm}k)}$$ is continuous for $\lm \in [-\frac 1{k-1}, 1]$, and takes value $0$ at $\lm = 1$;
			\item $m \lm f(x) f(\lm x + \frac{1-\lm}k)\ge 0$ for $\lm \in [-\frac 1{k-1}, 1]$.
		\end{enumerate}
		So we only need to prove that $$\frac{\der}{\der \lm} \frac{G(P)}{m \lm f(x) f(\lm x + \frac{1-\lm}k)} \le 0,$$ i.e.,
		$$f(\lm x + \frac{1-\lm}k)^2 \le \frac{k \lm^2 (kx-1)}{k-m} \frac{\der}{\der \lm} f(\lm x + \frac{1-\lm}k).$$
		Let $y = \lm x + \frac{1-\lm}k$. Then the above inequality can be rewritten as
		\begin{align*}
			f(y)^2 \le \frac{k \lm^2 (kx-1)}{k-m} \frac{\der y}{\der \lm}\frac{\der f(y)}{\der y}  = \frac{(ky-1)^2}{(k-m)y(1-my)}.
		\end{align*}
		This is true by Lemma \ref{LemmaElem}, applied to $a = \frac{ky-1}{1-my}$.
		Equality holds only when $y = \frac 1k$, which cannot happen for $\lm\ne 0$.
	\end{proof}

	The set of $(m, x)$ where $m\in [1, k-1]$, $x\in (\frac 1k, \frac 1m]$, and $H(P) = c$ can be parametrized as a curve
	$(m, x=x(m))$ for $m\in [1, m_c]$ for some constant $m_c$.
	Along the curve, $H(\PC_\lm \circ P)$ is continuous, and by Claim \ref{ClaimSDPIDer}, is increasing in $m$.
	Therefore $H(\PC_\lm \circ P)$ is minimized at $m=1$. This finishes the proof.
\end{proof}
\begin{proof}[Proof of Theorem \ref{ThmSDPIPotts}]
	Consider a Markov chain $U\to X\to Y$ where $X$ has uniform distribution, and the channel $X\to Y$ is $\PC_\lm$.
	Because $P_X$ and $P_Y$ are both uniform, for any $u$, we have
	\begin{align*}
		D(P_{X|U=u} || P_X) &= \log k - H(P_{X|U=u}),\\
		D(P_{Y|U=u} || P_Y) &= \log k - H(P_{Y|U=u}).
	\end{align*}
	So by Proposition \ref{PropSDPIPotts} we get
	$$D(P_{Y|U=u} || P_Y) \le s_\lm(D(P_{X|U=u} || P_X)).$$
	Therefore
	\begin{align*}
		I(U; Y) = D(P_{Y|U} || P_Y | P_U) \le \hat s_\lm(D(P_{X|U} || P_X | P_U)) = \hat s_\lm(I(U; X)).
	\end{align*}

	Now we prove optimality. Let $c\in [0, \log k]$.
	Choose $a, b\in [0, \log k]$ and $u\in [0, 1]$ such that
	$c=(1-u)a+u b$ and $\hat s_\lm(c) = (1-u) s_\lm(a) + us_\lm(b)$.
	Choose $\rho, \tau\in [0, 1]$ such that $\Capa(\PC_\rho)=a$ and $\Capa(\PC_\tau)=b$,
	where $\Capa$ denotes channel capacity.
	Define random variable $U = (V, Z)$ such that $Z \sim \Ber(u)$,
	and conditioned on $Z=0$, $V \sim \PC_{\rho}(X)$,
	and conditioned on $Z=1$, $V \sim \PC_{\tau}(X)$.
	One can check that $$I(U; X) = (1-u) a + u b = c$$
	and $$I(U; Y) = (1-u) s_\lm(a) + u s_\lm(b) = \hat s_\lm(c).$$
\end{proof}


Let us discuss the relationship between Potts semigroup and ferromagnetic Potts channels.
As discussed in the Introduction, ferromagnetic Potts channels are exactly the operators in the Potts semigroup,
with $T_t = \PC_{\exp(-\frac k{k-1} t)}$.
Therefore $1$-LSI for the Potts semigroup can be seen as infinitesimal SDPI for ferromagnetic Potts channels,
and many results for the former can directly transfer to results for the latter.

We use Proposition \ref{PropNonLin1LSI} to give an alternative proof for Proposition \ref{PropSDPIPotts}
for ferromagnetic Potts channels.

\begin{proof}[Alternative proof of Proposition \ref{PropSDPIPotts} for ferromagnetic Potts channels]
	Let $P$ and $Q$ be two distributions with $H(P) = H(Q) = c$, where $P$ is of form $(x, \frac {1-x}{k-1}, \ldots, \frac{1-x}{k-1})$ for some $x\in [\frac 1k, 1]$, and $Q$ is not of this form (up to permuting the alphabet).
	Define $P_t = T_t \circ P$ and $Q_t = T_t \circ Q$, where $(T_t)_{t\ge 0}$ is the Potts semigroup.

	We prove that $H(P_t) < H(Q_t)$ for $t\in (0, \infty)$.
	Suppose this does not hold.
	Let $u = \inf\{t > 0 : H(P_t) \ge H(Q_t)\}$.
	Then we have $H(Q_u) = H(P_u)$ by continuity of semigroup.
	By Proposition \ref{PropNonLin1LSI}, we have
	\begin{align*}
		\frac{\der}{\der t}|_{t=u} H(Q_t) = \cE(Q_u, \log Q_u) > \cE(P_u, \log P_u) = \frac{\der}{\der t}|_{t=u} H(P_t).
	\end{align*}
	If $u = 0$, then for some $\ep>0$, $H(Q_t) > H(P_t)$ for $t\in (0, \ep)$.
	If $u > 0$, then for some $\ep>0$, $H(Q_t) < H(P_t)$ for $t\in (u-\ep, u)$.
	Both cases lead to contradiction with definition of $u$.
	So $H(P_t) < H(Q_t)$ for $t\in (0, \infty)$. This completes the proof of the result for $\lambda > 0$.
\end{proof}

\subsection{Behavior for $k\to\infty$}\label{sec:klim}
When should one use $p$-NLSI instead of $p$-LSI? To get some insights, we consider the case of $k\to\infty$.
First, we prove Proposition~\ref{Prop1LSCPotts} that $\alpha_1 = 1 + \frac {1+o(1)}{\log k}$.
\begin{proof}[Proof of Proposition \ref{Prop1LSCPotts}]
	\textbf{Lower bound.}
	By Theorem \ref{ThmNonLinpLSI}, we need to show that
	for all $x\in (\frac 1k, 1]$, we have
	$$ \frac k{k-1} (1 + \frac 1{\log k}) \le \frac{\xi_1(x)}{\psi(x)}.$$
	Noting that $$\xi_1(x) = \frac k{k-1} ( \frac 1k(-\log x - (k-1) \log \frac{1-x}{k-1}) - \log k + \psi(x)),$$
	it suffices to prove that
	\begin{align*}
		f(x) := \frac {\log k}k (-\log x - (k-1) \log \frac{1-x}{k-1}) - \log^2 k - \psi(x) \ge 0.
	\end{align*}
	We have $f(\frac 1k)=0$.
	So it suffices to prove that $f^\p(x)\ge 0$ for $x\in [\frac 1k, 1]$.
	\begin{align*}
		f^\p(x) &= \frac{\log k}k (-\frac 1x + \frac {k-1}{1-x}) - (\log x - \log \frac{1-x}{k-1}).
	\end{align*}
	We smoothly continue this function to $\{(k,x)\in \bR^2 : k\ge 3, x\in [\frac 1k, 1]\}$
	and prove that it is non-negative in this region.
	\begin{align*}
		\frac{\der}{\der k} f^\p(x) &= \frac{1-\log k}{k^2} (-\frac 1x + \frac {k-1}{1-x})
		+ \frac{\log k}k \frac 1{1-x} - \frac 1{k-1}\\
		& = \frac{(k-1)\log k + k(kx^2 -x-1) +1}{k^2(k-1)x(1-x)}.
	\end{align*}
	The numerator is a quadratic function in $x$, and for fixed $k$, it is minimized at $x=\frac 1k$,
	leading to
	\begin{align*}
		\frac{\der}{\der k} f^\p(x) \ge \frac{(k-1)\log k-k+1}{k^2 (k-1)x(1-x)}
		= \frac{\log k-1} {k^2 x(1-x)} \ge 0.
	\end{align*}
	So we only need to prove $f^\p(x)\ge 0$ for minimum $k$, i.e., $k = \max\{3, \frac 1x\}$.
	When $k = \frac 1x$, on can verify that $f^\p(x)=0$.
	So the only remaining case is $k=3$.
	For $k=3$, we prove that $f$ is convex in $x$, i.e., $f^\pp(x)\ge 0$ for $x\in [0, 1]$.
	\begin{align*}
		f^\pp(x) &= \frac{\log k}k (\frac 1{x^2} + \frac {k-1}{(1-x)^2}) - (\frac 1x + \frac 1{1-x})\\
		& = \frac{\log k((1-x)^2 + (k-1)x^2) - k x(1-x)}{k x^2(1-x)^2}\\
		& = \frac{(k\log k + k) x^2 - (k+2\log k)x + \log k}{k x^2(1-x)^2}.
	\end{align*}
	The numerator is a quadratic function in $x$, and its discriminant is
	\begin{align*}
		(k+2\log k)^2 - 4 (k\log k+k) \log k = k^2 - 4(k-1)\log^2 k.
	\end{align*}
	When $k=3$, the above value is $<0$.
	So $f^\pp(x)\ge 0$ for $k=3$ and $x\in [0,1]$. This finishes the proof of the lower bound.

	\textbf{Upper bound.}
	For the upper bound, we need find $x\in (\frac 1k, 1]$ such that $\frac{f(x)}{\psi(x)} = o(1)$.
	Because the upper bound to prove is asymptotic, we assume that $k$ is large enough.
	Take $x = \frac 2{\log k}$.
	Then we have
	\begin{align*}
		\psi(x) &= \log k + \frac 2{\log k} \log \frac {2}{\log k} + (1-\frac 2{\log k}) \log \frac{1-\frac{2}{\log k}}{k-1}\\
		& = \log k - (1-\frac 2{\log k}) \log (k-1) + o(1)\\
		& = 2 +o(1)
	\end{align*}
	and
	\begin{align*}
		f(x) &= \frac{\log k}k(-\log \frac{2}{\log k} - (k-1) \log \frac{1-\frac{2}{\log k}}{k-1}) - \log^2 k -\psi(x)\\
		& =\frac{\log k}k(k-1) (\log (k-1) - \log (1-\frac{2}{\log k})) - \log^2 k - 2 +o(1)\\
		& = \log k \cdot (1+O(\frac 1k))\cdot (\log k + O(\frac 1k) + \frac{2}{\log k} + O(\frac 1{\log^2 k})) - \log^2 k-2+o(1)\\
		& = o(1).
	\end{align*}
	So $\frac{f(x)}{\psi(x)} = o(1)$.
\end{proof}
\begin{rmk}
	Numerical computation suggests $\frac{f(x)}{\psi(x)}$ is minimized at a point $x=\frac{2+o(1)}{\log k}$.
	This guides our proof of the upper bound in Proposition \ref{Prop1LSCPotts},
	but we have not attempted to prove this fact.
\end{rmk}

To understand the case $p>1$, let us denote convexification of $b_p$ as $\check b_p$. Then
NLSI lower bound, assuming $\EE[f]=1$, gives
$$\cE(f^{\frac1p}, f^{1-{\frac 1 p}}) \ge \check b_p(\Ent(f))\,.$$
We see that this improves upon $\alpha_p \cdot \Ent(f)$ the more the larger
the entropy. In particular, the maximum improvement happens when $\Ent(f) = \log k$. That is we have for $p>1$
$$ \alpha_p \le \frac{\check b_p(x)} x \le \frac 1{\log k}\,.$$
Together with \eqref{EqnpLSIComp} and \eqref{eq:2lsi_potts}, we get $\al_p = \Th(\frac 1{\log k})$ as $k\to \infty$.
Numerical computation suggests that $\al_p = \frac{1+o(1)}{\log k}$.

At the same time, the improvement given by the $1$-NLSI (over $1$-LSI) is much stronger, since
$b_1(\log k) = \infty$. \textit{To summarize,} the $p$-NLSI should be
preferred for $p=1$ or for cases where $k$ is small and entropy is large (i.e. functions are highly spiky).

Next, we consider SDPIs and $\eta_{\KL}$. First, we show that for a fixed
$\lambda\ge 0$ we have
$$\eta_{\KL}(\PC_\lambda,\pi) = {\lambda} - \Th\left(\frac 1{\log k}\right)\,.$$
Indeed, the upper bound is given by~\eqref{eq:etakl_1lsi}. For the lower bound we have
\begin{align*}
	& \eta_{\KL} (\PC_\lm, \pi)\\
	&\ge \eta_{\min} := \frac{\psi(\lm + \frac {1-\lm}k)}{\psi(1)} \\
	& = \frac{\log k + (\lm + \frac {1-\lm}k) \log (\lm + \frac {1-\lm}k) + (1-(\lm + \frac {1-\lm}k)) \log
	\frac{1-(\lm + \frac {1-\lm}k)}{k-1}}{\log k}\\
	& = \lm + \frac{\lambda \log \lambda + (1-\lambda)\log(1-\lambda) + o(1)}{\log k}\,.
\end{align*}

On the other hand,
\begin{align*}
	\eta_{\KL}(\PC_\lm, \pi) \le \eta_{\KL}(\PC_\lm) \le \eta_{\TV}(\PC_\lm) = \lm\,,
\end{align*}
where $\eta_{\TV}$ is the contraction coefficient for the total variation distance (Dobrushin coefficient,
see~\cite{CKZ98,PW17}).

Notice also that for $\hat s_\lambda$ we have generally $\eta_{\min} \le \frac{\hat s_\lambda(x)}x \le
\eta_{\KL}(\PC_\lambda,\pi)$. Therefore we have shown that
$$ \lim_{k\to \infty} \frac{\hat s_\lambda(x)}x = \lim_{k\to \infty}\eta_{\KL}(\PC_\lambda,\pi)
=\lim_{k\to \infty} \eta_{\KL}(\PC_\lambda) =
\eta_{\TV}(\PC_\lambda) = \lambda\,.$$
The estimates of information quantities using the more sophisticated tools get improvement over
simplistic coupling of at most multiplicative order $(1+\Th(\frac 1 {\log k}))$.

Note, however, if $\lm$ changes with $k$ (e.g.~$\lm = -\frac 1{k-1}$), then the improvement over
$\eta_{\TV}$ can be as large as a multiplicative factor of $(1+o(1))\log k$,
as shown in Proposition \ref{PropColorSDPI}.

\section{Product spaces}\label{SecProd}
In this section we study extensions of $p$-NLSIs and SDPIs to the
product semigroup $(T_t^{\ot n})_{t \ge 0}$ on the product space $[k]^n$ (and product channels $\PC_\lm^{\otimes n}$).
The general property of tensorization of $p$-NLSI was established in~\cite[Theorem 1]{PS19}, and thus we only need to
concavify functions $\Phi_p$ in~\eqref{eq:nlsi}. Similarly, we can show that (non-linear) strong data processing
inequalities tensorize if one concavifies function $s(\cdot)$ in~\eqref{eq:nsdpi}.

After showing these extensions to product spaces, we proceed to discussing implications of $p$-NLSI on speed of
convergence to equillibrium in terms of $\Ent(T_t^{\otimes n} \nu)$ and on edge-isoperimetric inequalities.

\subsection{Tensorization}

\begin{prop}\label{PropLSITenso}
	Fix $p\ge 1$.
	Recall $b_p$ defined in \eqref{eq:bp_def}.
	Let $\check b_p$ be the convex envelope of $b_p$.
	Then $p$-LSI holds for the product semigroup $(T_t^{\ot n})_{t \ge 0}$
	with $$\Phi_{n, p}(x) = n \check b_p^{-1}(\frac xn).$$
\end{prop}
\begin{proof}
	By Theorem \ref{ThmNonLinpLSI} and \cite[Theorem 1]{PS19}.
\end{proof}
As we show below $\check b_p \neq b_p$ (Proposition~\ref{PropNonConv}).

For non-linear SDPI, we first prove a general tensorization result.
\begin{prop}\label{PropSDPITensoGen}
	Fix a probability kernel $P_{Y|X}: \cX \to \cY$ and a distribution $P_X$ on $\cX$.

	\begin{enumerate}
		\item Suppose for some non-decreasing function $s: \bR_{\ge 0} \to \bR_{\ge 0}$ we have
		\begin{equation}
			D(Q_Y || P_Y) \le s(D(Q_X||P_X))\label{EqnNonLinKLSDPI}
		\end{equation}
		for all distribution $Q_X$ on $\cX$ with $0< D(Q_X||P_X)<\infty$.
		Then for all distribution $Q_{X^n}$ on $\cX^n$ with $0 < D(Q_{X^n} || P_X^{\ot n}) < \infty$,
		we have
		\begin{equation}
			D(Q_{Y^n} || P_Y^{\ot n}) \le n \hat s(\frac 1n D(Q_{X^n} || P_X^{\ot n})),
		\end{equation}
		where $\hat s$ is the concave envelope of $s$,
		and $Q_{Y^n} = P_{Y|X}^{\ot n} \circ Q_{X^n}$.
		\item Suppose for some non-decreasing concave function $\hat s: [0, \log |\cX|]\to \bR_{\ge 0}$ we have
		\begin{equation}
			I(U; Y) \le \hat s(I(U; X))\label{EqnNonLinMISDPI}
		\end{equation}
		for all Markov chains
		$$U\to X\to Y$$
		where the distribution of $X$ is $P_X$.
		Then for all Markov chains
		$$U\to X^n\to Y^n$$
		where the distribution of $X$ is $P_X^{\ot n}$,
		we have
		\begin{equation}
			I(U; Y^n) \le n \hat s(\frac 1n I(U; X^n)).
		\end{equation}
	\end{enumerate}
\end{prop}
We have separate statements for non-linear SDPI defined via KL divergence and via mutual information,
because they are not equivalent in general.
It is not hard to show that if KL divergence type non-linear SDPI (Inequality \eqref{EqnNonLinKLSDPI}) holds for some function $s$,
then mutual information type non-linear SDPI (Inequality \eqref{EqnNonLinMISDPI}) holds for $\hat s$.
However, it is not clear what is the best possible KL divergence type non-linear SDPI one can get
starting from mutual information type non-linear SDPI.
(Note the domain of function $s$ would become larger during the translation.)
\begin{proof}[Proof of Proposition \ref{PropSDPITensoGen}]
	\textbf{Proof of (1).}
	Perform induction on $n$. The base case $n=1$ is trivial.
	Now consider $n\ge 2$.
	We have
	\begin{align*}
		& D(Q_{Y^n} || P_Y^{\ot n}) \\
		&= D(Q_{Y^{n-1}} || P_Y^{\ot (n-1)}) + D(Q_{Y_n | Y^{n-1}} || P_Y | Q_{Y^{n-1}})\\
		& \le D(Q_{Y^{n-1}} || P_Y^{\ot (n-1)}) + D(Q_{Y_n | X^{n-1}} || P_Y | Q_{X^{n-1}}) \\
		& \le (n-1) \hat s(\frac 1{n-1} D(Q_{X^{n-1}}|| P_X^{\ot (n-1)})) + s(D(Q_{X_n| X^{n-1}} || P_X | Q_{X^{n-1}}))\\
		& \le n \hat s(\frac 1 n D(Q_{X^{n-1}}|| P_X^{\ot (n-1)}) + \frac 1n D(Q_{X_n| X^{n-1}} || P_X | Q_{X^{n-1}}))\\
		& = n \hat s(\frac 1n D(Q_{X^n} || P_X^{\ot n})).
	\end{align*}
	First step is by chain rule.
	Second step is because we have a Markov chain $$Y^{n-1} - X^{n-1} - Y_n,$$
	and conditioning on more information does not decrease conditioned divergence.
	Third step is by induction hypothesis.
	Fourth step is by concavity.
	Fifth step is by chain rule.

	\textbf{Proof of (2).}
	Perform induction on $n$. The base case $n=1$ is trivial.
	Now consider $n\ge 2$.
	We have
	\begin{align*}
		I(U; Y^n) &= I(U; Y^{n-1}) + I(U; Y_n | Y^{n-1}) \\
		& = I(U; Y^{n-1}) + I(U, Y^{n-1}; Y_n) \\
		& \le I(U; Y^{n-1}) + I(U, X^{n-1}; Y_n) \\
		& = I(U; Y^{n-1}) + I(U; Y_n | X^{n-1}) \\
		& \le (n-1) \hat s(\frac 1{n-1} I(U; X^{n-1})) + \hat s(I(U; X_n | X^{n-1})) \\
		& \le n \hat s(\frac 1n I(U; X^{n-1}) + \frac 1n I(U; X_n | X^{n-1})) \\
		& = n \hat s (\frac 1n I(U; X^n)).
	\end{align*}
	First step is by chain rule.
	Second step is by chain rule, and that $Y_n$ is independent with $Y^{n-1}$.
	Third step is by data processing inequality.
	Fourth step is by chain rule, and that $Y_n$ is independent with $X^{n-1}$.
	Fifth step is by induction hypothesis.
	Sixth step is by concavity.
	Seventh step is by chain rule.
\end{proof}

\begin{coro}\label{CoroSDPITenso}
	Recall function $s_\lm$ defined in Theorem \ref{ThmSDPIPotts}.
	Let $Q_{X^n}$ be a distribution on $[k]^n$ and $Q_{Y^n} = \PC_\lm^{\ot n}\circ Q_{X^n}$.
	Then we have
	\begin{equation}
		{\frac 1n} H(Y^n) \ge \log k - \hat s_\lm\left(\log k - {\frac 1n}H(X^n)\right)\,.\label{EqnSDPITenso}
	\end{equation}
	Furthermore, for every $c\in [0,\log k]$, there exist distributions $X^n$ with $H(X^n)=(c+o(1))n$
	such that ${\frac 1n} H(Y^n) = \log k - \hat s_\lm(\log k - c) + o(1)$.
\end{coro}
\begin{proof}


	Inequality \eqref{EqnSDPITenso} follows from Proposition \ref{PropSDPITensoGen}
	and that $$D(Q_{X^n} || \pi^{\ot n}) = n\log k - H(Q_{X^n}).$$

	For the second part, choose $a, b\in [0, \log k]$ and $u \in [0, 1]$ such that
	$c = (1-u) a + ub$ and $\hat s_\lm (\log k - c) = (1-u) s_\lm(\log k - a) + u s_\lm(\log k - b)$.
	Such $a, b, u$ exist because $\hat s_\lm$ is the concave envelope of $s_\lm$.

	Let $Q_A$ (resp. $Q_B$) be the unique distribution on $[k]$ of form $(x, \frac {1-x}{k-1}, \ldots, \frac{1-x}{k-1})$ with $x\in [\frac 1k, 1]$ and entropy $a$ (resp.~entropy $b$).
	Now let $Q_{X^n}$ be the distribution
	$Q_A \t \cdots \t Q_A \t Q_B \t \cdots \t Q_B$, where
	$Q_A$ appears $\lfloor (1-u) n\rfloor$ times and $Q_B$ appears $\lceil u n\rceil$ times.
	It is easy to see that this distribution satisfies the required properties.
\end{proof}

\subsection{Linear piece}\label{sec:linpiece}
In Proposition \ref{PropLSITenso} and Theorem \ref{ThmSDPIPotts}, we make use of convexification
of $b_p$ and concavification of $s_\lm$.
When $k=2$, we have $\check b_p = b_p$ and $\hat s_\lm = s_\lm$ (the latter fact is known as Mrs. Gerber's Lemma \cite{WZ73}).
However, for $k\ge 3$, the situation is vastly different.
\begin{prop}\label{PropNonConv}
	Recall function $b_p: [0, \log k]\to \bR$ defined in Theorem \ref{ThmNonLinpLSI} and $s_\lm: [0, \log k] \to \bR$ defined in Theorem \ref{ThmSDPIPotts}.
	\begin{enumerate}
		\item For all $k\ge 3$ and $p\ge 1$, $b_p$ is not convex near $0$.
		\item For all $k\ge 3$, $\lm \in [-\frac 1{k-1}, 0)\cup (0, 1)$, $s_\lm$ not concave near $0$.
	\end{enumerate}
\end{prop}
The proof is deferred to Appendix \ref{SecNonConv}.
Proposition \ref{PropNonConv} implies that there is a linear piece near origin in the graph of
$\check b_p$, $\hat \Phi_p$ and $\hat s_\lm$.

This implies a curious new property distinguishing Potts
semigroup with $k\ge 3$ from its binary cousin and from the Ornstein-Uhlenbeck semigroup.
Both of the latter have their $p$-NLSI
and SDPI strictly non-linear, which translates into the
following fact: among all initial densities $\nu_0$ with a given entropy
$\Ent(\nu_0)$ a simple product distributions simultaneously maximizes $\Ent(T_t^{\otimes n}\nu_0)$ for all $t$.
Stated differently we have (this is known as Mrs.~Gerber's Lemma) when $k=2$:
\begin{equation}
	D(\PC_\lm \circ P_{X^n}||\pi^{\ot n}) \le D(\PC_\lm \circ \Ber^{\otimes n}(p) ||\pi^{\ot n}),\label{EqnMGLk2}
\end{equation}
where $\pi$ is the uniform distribution on $[k]^n$,
and $\Ber^{\otimes n}(p)$ is an iid distribution on $[k]^n$ with $p\in[0,1/2]$ solving $D(\Ber(p)\|\pi) = \frac 1n
D(\PC_\lm \circ P_{X^n}||\pi^{\ot n})$. \textit{That is, the slowest to relax to equillibrium is the product
distribution.} For the Ornstein-Uhlenbeck a similar statement holds with $\Ber(p)$ replaced by the
$\cN(0,\sigma^2 I_n)$.

This nice extremal property of product distributions is no longer true for $k\ge 3$ Potts semigroups, because $s_\lm$ is not concave, and
the value of $\hat s_\lm$ at a point may be a mixture of two values of $s_\lm$. More precisely, instead of \eqref{EqnMGLk2}, we have
for every $\lm \in [-\frac 1{k-1}, 1]$ and every $c \in \bR_{\ge 0}$, there exist two iid distributions
$P_1, P_2$ on $[k]^n$ and $t\in [0, 1]$ satisfying
\begin{align*}
	(1-t) D(P_1 ||\pi^{\ot n}) + t D(P_2 ||\pi^{\ot n})=c,
\end{align*}
such that for every distribution $P_{X^n}$ on $[k]^n$ with $D(P_{X^n} || \pi^{\ot n}) = c$,
we have
\begin{align*}
	D(\PC_\lm \circ P_{X^n}||\pi^{\ot n}) \le (1-t) D(\PC_\lm \circ P_1 ||\pi^{\ot n})
	+ t D(\PC_\lm \circ P_2 ||\pi^{\ot n}).
\end{align*}
Note here $P_1$, $P_2$ and $t$ all depend on $c$ and $\lm$, and thus there is no universal distribution that is the
slowest to converge to equillibrium.

Let us discuss some general implications of non-convexity of $b_p$ and non-concavity of $s_\lm$ near $0$.

Let $K$ be a Markov kernel with stationary distribution $\pi$.
Consider the tightest possible $p$-NLSI given by
$$b_p(x) := \inf_{\substack{f: \cX \to \bR_{\ge 0},\\ \bE_\pi f=1, \Ent_\pi(f)=x}} \cE(f^{\frac 1p}, f^{1-\frac 1p}).$$
The $p$-log-Sobolev constant is
$$\al_p := \inf_{x>0} \frac{b_p(x)}x = \inf_{f: \cX \to \bR_{\ge 0}, \Ent_\pi(f)>0} \frac{\cE(f^{\frac 1p}, f^{1-\frac 1p})}{\Ent_\pi(f)}.$$
We also define the spectral gap
$$\lm := \inf_{f: \cX \to \bR_{\ge 0}, \Var(f)>0} \frac{\cE(f, f)}{\Var(f)},$$
where $\Var(f) = \bE_\pi (f-\bE_\pi f)^2$.
For any $p>1$, we have
\begin{align}\label{EqnpLSCvsSpec}
	\frac{p^2}{2(p-1)} \al_p \le \lm.
\end{align}
The case $p=2$ is proved in Diaconis and Saloff-Coste \cite{DSC96},
and  the general case is proved in Mossel et al.~\cite{MOS13}.
Their proof in fact implies a stronger inequality.
\begin{lemma}
	$$\limsup_{x\to 0^+}\frac{b_p(x)}x \le \frac{2(p-1)}{p^2}\lm.$$
	In particular, when $b_p$ is strictly concave near $0$, we have
	$$\al_p < \limsup_{x\to 0^+}\frac{b_p(x)}x \le \frac{2(p-1)}{p^2}\lm,$$
	and \eqref{EqnpLSCvsSpec} is strict.
\end{lemma}
\begin{proof}
	Take any $g: \cX\to \bR_{\ge 0}$ with $\Var(g)>0$. Define $f_\ep = 1+\ep g$.
	As $\ep\to 0$, we have
	\begin{align*}
		\cE(f_\ep^{\frac 1p}, f_\ep^{1-\frac 1p}) &= \ep^2 \frac 1p (1-\frac 1p) \cE(g, g) + o(\ep^2),\\
		\Ent_\pi(f_\ep) &= \frac 12 \ep^2 \Var(g) + o(\ep^2).
	\end{align*}
	Because $\Ent_\pi(f_\ep)\to 0$ continuously as $\ep\to 0$, we have
	\begin{align*}
		\limsup_{x\to 0^+}\frac{b_p(x)}x
		\le \lim_{\ep\to 0} \frac{\cE(f_\ep^{\frac 1p}, f_\ep^{1-\frac 1p})}{\Ent_\pi(f_\ep)}
		= \frac{2(p-1)}{p^2}\frac{\cE(g, g)}{\Var(g)}.
	\end{align*}
	Lemma then follows because $g$ is arbitrary.
\end{proof}
Roughly speaking, existence of a ``linear piece'' near $0$ in $\check b_p$ implies that \eqref{EqnpLSCvsSpec} is strict.
For the Potts semigroup with $k\ge 3$, $b_p$ is strictly concave near $0$ by proof of Proposition \ref{PropNonConv}.
So \eqref{EqnpLSCvsSpec} is strict for the Potts semigroup.

The story for non-linear SDPI is very similar.
Let $W$ be any channel and $Q$ be any input distribution.
Consider the tightest possible non-linear SDPI given by
\begin{align*}
	s(x) := \sup_{P : D(P||Q) = x} D(PW|| QW).
\end{align*}
The input-restricted KL divergence contraction coefficient is
\begin{align*}
	\eta_{\KL}(W, Q) := \sup_{x>0} \frac{s(x)}x = \sup_{P: 0<D(P||Q)<\infty} \frac{D(PW||QW)}{D(P||Q)}.
\end{align*}
We also consider the input-restricted $\chi^2$-divergence contraction coefficient
\begin{align*}
	\eta_{\chi^2}(W, Q) := \sup_{P: 0<\chi^2(P||Q)<\infty} \frac{\chi^2(PW||QW)}{\chi^2(P||Q)}.
\end{align*}
It is known (Ahlswede and G\'{a}cs \cite{AG76}) that
\begin{align}\label{EqnetaKLvsetachi2}
	\eta_{\KL}(W, Q) \ge \eta_{\chi^2}(W, Q).
\end{align}
Similarly to the $p$-NLSI case, the proof of \eqref{EqnetaKLvsetachi2} implies a stronger inequality.
\begin{lemma}
	$$\liminf_{x\to 0^+} \frac{s(x)}x \ge \eta_{\chi^2}(W, Q).$$
	In particular, when $s$ is strictly convex near $0$,
	we have $$\eta_{\KL}(W, Q) > \liminf_{x\to 0^+} \frac{s(x)}x \ge \eta_{\chi^2}(W, Q).$$
	and \eqref{EqnetaKLvsetachi2} is strict.
\end{lemma}
\begin{proof}
	Fix any distribution $P$ with $0 < \chi^2(P||Q) < \infty$.
	Proof of \cite[Theorem 2]{PW17} constructs a sequence of distributions
	$P_\ep$ satisfying
	\begin{align*}
		D(P_\ep || Q) &= \ep^2 \chi^2 (P||Q)+ o(\ep^2),\\
		D(P_\ep W || Q W) &= \ep^2 \chi^2 (P W||Q W)+ o(\ep^2),
	\end{align*}
	and $D(P_\ep || Q)\to 0$ continuously as $\ep \to 0$.
	Therefore
	\begin{align*}
		\liminf_{x\to 0^+} \frac{s(x)}x
		\ge \lim_{\ep \to 0} \frac{D(P_\ep W || Q W)}{D(P_\ep || Q)}
		= \frac{\chi^2(PW||QW)}{\chi^2(P||Q)}.
	\end{align*}
	Lemma follows because $P$ is arbitrary.
\end{proof}
Roughly speaking, existence of a ``linear piece'' near $0$ in $\hat s$ implies that \eqref{EqnetaKLvsetachi2} is strict.
For Potts channels $\PC_\lm$ with $\lm \in [-\frac 1{k-1}, 0) \cup (0, 1)$ and $k\ge 3$,
$s_\lm$ is strictly convex near $0$ by proof of Proposition \ref{PropNonConv}.
So \eqref{EqnetaKLvsetachi2} is strict for Potts channels.


\subsection{Edge isoperimetric inequalities}\label{sec:edgeisop}
As a toy application of the NLSIs for the product spaces, we derive an edge isoperimetric
inequality for $K_k^n$, the graph whose vertex set is $[k]^n$, and edges connect vertex pairs with Hamming distance one.
Given a graph $G=(V, E)$, edge isoperimetric inequalities solve the following combinatorial optimization problem:
$$ \Psi_G(N) = \min \{|E(S,S^c)|: |S| = N\}\,,$$
where $|E(S, S^c)| = \#\{e \in E: |e\cap S| = 1\}$.
For $K_k^n$, the edge isoperimetric problem has been completely solved~\cite{Har64,Lin64,Ber67,Har76}.
Specifically, Lindsey \cite{Lin64} showed that the optimal $S$ minimizing $|E(S,S^c)|$ for a fixed $|S|$ consists of
largest elements in $[k]^n$ under a lexicographical order. In particular, we have
$$ \Psi_{K_k^n}(k^m) = (n-m) (k-1) k^{m}\,. $$
This was obtained by an explicit combinatorial argument (via a form of shifting/compression).
What estimates obtained via LSIs and NLSIs?

Let $f = \bbl_S$ be the indicator function of a set $S$. Then for any $p>1$ we have
$$\frac{\cE(f^p, f^{1-p})}{\bE_\pi[f]} = \frac 1{k-1} \frac{|E(S, S^c)|}{|S|} \qquad \text{and} \qquad
\frac{\Ent(f)}{\bE_\pi[f]} = \log \frac{k^n}{|S|}\,.$$

If we relate these two ratios via the $2$-LSI (note that from \eqref{EqnpLSIComp}, of all $p>1$ the $p=2$ gives the
best result here) and by using the known value of $\alpha_2$ from~\eqref{eq:2lsi_potts} we get
\begin{equation}\label{eq:2lsi_isop}
	\Psi_{K_k^n}(k^m) \ge k^m (n-m) (k-2) \frac{\log k}{\log(k-1)}\,.
\end{equation}
Clearly the coefficient in front of $(n-m)k^m$ here is not tight.

The $p$-NLSI allows us to perform a better comparison. First, again via \eqref{EqnpLSIComp} we get the best
inequality for $p=2$, which results in
\begin{equation}
	\Psi_{K_k^n}(k^m)\ge (k-1) k^m n \check b_2(\frac{n-m}{n} \log k).\label{EqnLSIIsop}
\end{equation}
We know that the function $\check b_2$ is continuous with $\check b_2(\log k) = b_2(\log k) = 1$
(from~\eqref{eq:bp_def}). Thus, for any $m=o(n)$ and $n\to\infty$ we get that~\eqref{EqnLSIIsop} implies
$$ \Psi_{K_k^n}(k^m) \ge (k-1) k^m (n-m) (1+o(1))\,, $$
which is tight in this regime. (However, from~\eqref{eq:bp_def} we can also find that $\check b_2'(1)=\infty$ and thus,
even when $m = o(n)$ the right-hand side of the above inequality is $(k-1)k^m (n-\omega(m))$, implying the behavior in
terms of $m$ is not optimal.)

\section{Non-reconstruction for broadcast models on trees}\label{SecTreeGeneral}
In this section we prove non-reconstruction results for a general class of broadcast models on trees, using input-restricted KL divergence SDPI.

Fix a channel $M : [k] \to [k]$ with an invariant distribution $q^*$, i.e., $q^* M = q^*$.
Let $M^\vee$ denote the reverse channel, i.e., any channel that satisfies $q^*_j M^\vee_{j,i} = q^*_i M_{i,j}$ for all $i,j\in [k]$.
Consider a (possibly infinite) tree $T$ with a marked root $\rho$. For each vertex $v\in T$, we generate a spin $\sm_v \in [k]$ according to the following rules:
\begin{enumerate}
	\item $\bP(\sm_\rho = i) = q^*_i$.
	\item If $u$ is the parent of $v$, then $\bP(\sm_v = j | \sm_u = i) = M_{i,j}$.
\end{enumerate}
Let $L_h$ denote the set of vertices of distance $h$ to $\rho$. We say the model has non-reconstruction if and only if
$$\lim_{h\to\infty} I(\sm_\rho; \sm_{L_h}) = 0.$$
For any vertex $u$, let $c(u)$ denote the set of children of $u$.

Finally, recall the definition of the branching number $\br(T)$ of a tree $T$ by Lyons \cite{Lyo90}.
\begin{defn}[Branching number]\label{DefnBranch}
	Define a flow to be a function $f: V(T) \to \bR_{\ge 0}$ such that
	for every vertex $u$, we have $$f_u = \sum_{v\in c(u)} f_v.$$
	Define $\br(T)$ to be the $\sup$ of all numbers $\lm$ such that there exists a flow $f$ with
	$f_\rho > 0$, and $f_u \le \lm^{-d(u,\rho)}$ for all vertices $u$, where $d(u, \rho)$ is the distance between $u$ and $\rho$.
\end{defn}

Recall Theorem \ref{ThmTreeGeneral} states that the model has non-reconstruction when $$\eta_{\KL}(M^\vee, q^*) \br(T) < 1.$$
Now we prove the theorem.

\begin{proof}[Proof of Theorem \ref{ThmTreeGeneral}]
	For any vertex $u$, define $L_{u,h}$ to be the set of descendants of $u$ that have distance $h$ to $\rho$.
	Define $$a_u = H(q^*)^{-1}\eta_{\KL}(M^\vee, q^*)^{d(u,\rho)}\lim_{h\to \infty} I(\sm_u; \sm_{L_{u,h}}).$$
	By data processing inequality, $I(\sm_u; \sm_{L_{u,h}})$ is non-increasing for $h\ge d(u, \rho)$, so the limit exists.

	For any $v\in c(u)$, consider the reverse Markov chain
	$$\sm_{L_{v,h}} \to \sm_v \to \sm_u.$$
	Because $q^*$ is an invariant distribution, the distributions of $\sm_v$ and $\sm_u$ are both $q^*$.
	By SDPI, we have
	\begin{align*}
		I(\sm_u; \sm_{L_{v,h}}) \le \eta_{\KL}(M^\vee, q^*) I(\sm_v; \sm_{L_{v,h}}).
	\end{align*}
	Because $(\sm_{L_{v,h}})_{v\in c(u)}$ are independent conditioned on $\sm_u$, we have
	\begin{align*}
		I(\sm_u; \sm_{L_{u,h}}) \le \sum_{v\in c(u)} I(\sm_u; \sm_{L_{v,h}}).
	\end{align*}
	Combine the two inequalities and let $h\to \infty$. We get that
	$$a_u \le \sum_{v\in c(u)} a_v.$$

	Clearly, $$a_u \le \eta_{\KL}(M^\vee, q^*)^{d(u,\rho)}$$ for all vertices $u$. However, $a$ is not quite a flow yet.
	We define a flow $b$ from $a$.
	For a vertex $u$, let $u_0 = \rho, \ldots, u_d = u$ be the shortest path from $\rho$ to $u$.
	Define $$b_u = a_u \prod_{0\le j\le d-1} \frac{a_{u_j}} {\sum_{v\in c(u_j)} a_v}.$$
	(If for some $j$, we have $\sum_{v\in c(u_j)} a_v=0$, then let $b_u = 0$.)
	It is not hard to check that
	$$b_u = \sum_{v \in c(u)} b_v,$$ and that $$b_u \le a_u \le \eta_{\KL}(M^\vee, q^*)^{d(u,\rho)}.$$
	By definition of branching number, we must have $b_\rho = 0$.
	This means $$\lim_{h\to \infty} I(\sm_\rho; \sm_{L_h}) = 0,$$
	and non-reconstruction holds.
\end{proof}
\begin{rmk}
	In the definition of the model, it is not necessary to require $\sm_\rho$ to have distribution $q^*$.
	If we let $\sm^i_{L_h}$ denote the leaf colors conditioned on $\sm_\rho=i$, then Theorem \ref{ThmPotts} implies that
	when $\eta_{\KL}(M^\vee, q^*) \br(T) < 1$, we have
	$$\lim_{h\to\infty} \TV(\sm^i_{L_h}, \sm^j_{L_h}) = 0$$
	for $i\ne j$ with $q^*_i, q^*_j > 0$.
\end{rmk}
Theorem \ref{ThmTreeGeneral} directly implies non-reconstruction results for Galton-Watson trees.
\begin{coro}
	Let $T$ be a Galton-Watson tree with expected offspring $d$.
	If $$\eta_{\KL}(M^\vee, q^*)d < 1,$$ the model has non-reconstruction a.s.
\end{coro}
\begin{proof}
	If $T$ extincts, then non-reconstruction obviously hold.
	Conditioned on non-extinction, we have $\br(T) = d$ a.s.~by \cite{Lyo90}, and then Theorem \ref{ThmTreeGeneral} applies.
\end{proof}
\begin{rmk}\label{RmkComparisonGeneral}
	As shown in Polyanskiy and Wu~\cite{PW17} reconstruction problems on arbitrary directed acyclic graphs (in
	particular trees) can be shown to be non-reconstructible by reducing to a percolation problem on the same graph.
	For example for trees, this results in non-reconstruction as long as
	\begin{equation}\label{eq:tree_nonrec}
			\eta_{\KL}(M) \br(T) < 1.
	\end{equation}
	For any channel $M$, we have $$\eta_{\KL}(M,q^*) \le  \eta_{\KL}(M),$$ and the inequality is often strict. So
	for reversible channels (i.e., $M = M^\vee$), Theorem \ref{ThmTreeGeneral} implies result~\eqref{eq:tree_nonrec}. We do not know, however, how to extend Theorem~\ref{ThmTreeGeneral} to  general
	(non-tree) DAGs using input-restricted contraction coefficients.

	Formentin and K{\"u}lske \cite{FK09a} proved a non-reconstruction result very similar to ours.
	They considered the symmetrized KL divergence $$D_{\SKL}(P || Q) = D(P||Q) + D(Q||P),$$ which is $f$-divergence with $f(x) = (x-1)\log x$.
	They proved non-reconstruction holds for a Galton-Watson tree with expected offspring $d$ if
	$$\eta_{\SKL} (M^\vee, q^*) d < 1.$$
	With the method of proof as in Theorem~\ref{ThmTreeGeneral} one can strengthen their result so that non-reconstruction holds for $$\eta_{\SKL}(M^\vee, q^*) \br(T) < 1.$$

	Proceeding to input-restricted contraction coefficients, we computed both numerically for several binary asymmetric channels and Potts channels.
	In Figure \ref{FigBinAsm_SKL_vs_KL}, we compare $\eta_{\SKL}(M, q^*)$ and $\eta_{\KL}(M, q^*)$ for the binary asymmetric channel $$M = \begin{pmatrix} 1-a & a \\ b & 1-b\end{pmatrix},$$ for $a=0.3$ and $b\in [0, 1]$. Simple computation shows that $q^* = (\frac b{a+b}, \frac a{a+b})$ and $M = M^\vee$.
	\begin{figure}[ht]
		\centering
		\includegraphics[scale=0.4]{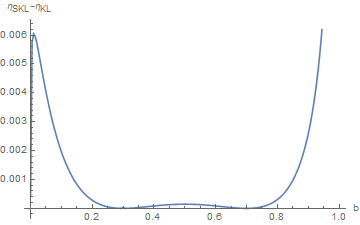}
		\caption{Contraction coefficient comparison for binary asymmetric channels with $a=0.3$ and varying $b \in [0, 1]$.\\
		The figure shows $\eta_{\SKL}(M, q^*) - \eta_{\KL}(M, q^*)$ is always non-negative.}
		\label{FigBinAsm_SKL_vs_KL}
	\end{figure}

	In Figure \ref{FigPotts_SKL_vs_KL}, we compare the input-restricted SKL and KL contraction coefficients for Potts channels $\PC_\lm$ for $k=5$ and $\lm \in [-\frac 1{k-1}, 1]$.
	Because a simplified expression for $\eta_{\SKL}(\PC_\lm, q^*)$ is not known, we use a lower bound $\ol \eta_{\SKL}(\PC_\lm, q^*)$, which is defined as the sup of $\frac{D_{\SKL}(\PC_\lm \circ P || q^*)}{D_{\SKL}(P||q^*)}$ considering only distributions $P = (p_1, \ldots, p_k)$ with $p_2 = \cdots = p_k$.
	Clearly $\ol \eta_{\SKL}(\PC_\lm, q^*) \le \eta_{\SKL}(\PC_\lm, q^*)$. It is conjectured in Formentin and K{\"u}lske \cite{FK09b} that $\ol \eta_{\SKL}(\PC_\lm, q^*) = \eta_{\SKL}(\PC_\lm, q^*)$ always.
	\begin{figure}[ht]
		\centering
		\includegraphics[scale=0.4]{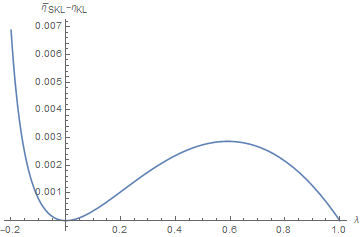}
		\caption{Contraction coefficient comparison for Potts channel with $k=5$ and varying $\lm \in [-\frac 1{k-1}, 1]$.\\
		The figure shows $\ol \eta_{\SKL}(\PC_\lm, q^*)-\eta_{\KL}(\PC_\lm, q^*)$ is always non-negative.}
		\label{FigPotts_SKL_vs_KL}
  \end{figure}

	As shown in Figure \ref{FigBinAsm_SKL_vs_KL} and Figure \ref{FigPotts_SKL_vs_KL}, in all these cases, we observe $$\eta_{\KL}(M^\vee, q^*) \le \eta_{\SKL}(M^\vee, q^*),$$ which means Theorem \ref{ThmTreeGeneral} yields a stronger non-reconstruction result for these cases.

	We remark that the input-unrestricted $\KL$ and $\SKL$ contraction coefficients agree.
	Indeed, the function $x \mapsto (x-1)\log x$ is operator convex (e.g.,~\cite[Example
	3.6]{chansangiam2013operator}), and thus by~\cite[Theorem 1]{CRS94}, we have
	\begin{align*}
		\eta_{\KL}(M) = \eta_{\SKL}(M)
	\end{align*}
	for any channel $M$.

	Suppose that for some function $g$, the $g$-mutual information satisfies the following subadditivity property: for any Markov chain $Y-X-Z$, we have
	$$I_g(X; Y, Z) \le I_g(X; Y) + I_g(X; Z).$$
	Then non-reconstruction holds for a tree $T$ with $$\eta_g(M^\vee, q^*) \br(T) < 1,$$
	by the proof of Theorem \ref{ThmTreeGeneral}.
	For mutual information the subadditivity is standard. For $I_{\SKL}$, Formentin-K{\"u}lske \cite{FK09b} proved that
	$$I_{\SKL}(X; Y, Z) = I_{\SKL}(X; Y) + I_{\SKL}(X; Z).$$
	It is an interesting question what is the best possible contraction coefficient one can achieve by varying $g$.
\end{rmk}

In Appendix~\ref{SecBOTGauss}, we study a broadcasting on trees model with Gaussian kernel considered in Eldan et al.~\cite{Eld20}, and prove tight non-reconstruction results for this model, closing a gap left in op.~cit.

\section{Potts model on a tree} \label{SecPotts}

In this section, we apply Theorem \ref{ThmTreeGeneral} to get non-reconstruction results for Potts models on a tree.
In the Potts model, spins propagate through the Potts channel $\PC_\lm$.
Because Potts channels are reversible, $\PC_\lm^\vee = \PC_\lm$. Because Potts channels are symmetric, the invariant distribution is $\pi = \Unif([k])$.
Thus Theorem \ref{ThmPotts} directly follows from Theorem \ref{ThmTreeGeneral}.

Let us briefly discuss previous non-reconstruction results for the Potts channel.
Mossel and Peres \cite{MP03} proved non-reconstruction for $$\frac{k\lm^2}{(k-2)\lm + 2} \br(T) < 1.$$
By Proposition \ref{PropConPotts} we can see, thus, that this exactly corresponds to invoking a weaker
(unrestricted-input) version of $\eta_{\KL}$.
Therefore, Theorem \ref{ThmPotts} is strictly stronger than \cite{MP03} (see also discussion in Remark
\ref{RmkComparisonGeneral}).
Martinelli et al.~\cite{MSW07} proved non-reconstruction for regular trees for $$d (1-\ep) \frac{k\lm^2}{(k-2)\lm + 2} < 1$$
where $\ep>0$ is a function of $k\ge 3$, $d$, $\lm$ in some involved way.
Sly \cite{Sly09a} obtained very sharp results for regular trees. In particular, he proved that Kesten-Stigum bound is tight for $k=3$ and large enough $d$. It looks hard to extract what general bounds one can achieve by Sly's method.
Formentin and K{\"u}lske \cite{FK09b} gave non-reconstruction results very similar to ours. As discussed in Remark \ref{RmkComparisonGeneral}, numerical computation suggests that Theorem \ref{ThmPotts} is stronger than their results.

For certain parameters, we can compute the contraction coefficient $\eta_{\KL}(\PC_\lm, \pi)$ in closed form. In the following we show two examples.
\subsection{Binary symmetric channel}
For $k=2$, $\PC_\lm$ is the binary symmetric channel $\BSC_\de$ with $\de = \frac{1-\lm}2$, which is known (Ahlswede and G\'{a}cs \cite{AG76}) to have SDPI coefficient $\eta_{\KL}(\BSC_\de) = (1-2\de)^2 = \lm^2$.
Theorem \ref{ThmPotts} implies non-reconstruction for $(1-2\de)^2 \br(T) < 1$, which was shown
in Bleher et al.~\cite{BRZ95} (for regular trees) and Evans et al.~\cite{EKPS00} (for general trees).

\subsection{Random coloring}
The random coloring model with $k$-colors corresponds to the channel $\Col_k :=
\PC_{-\frac{1}{k-1}}$. This channel acts on input $x\in[k]$ by outputting $y\neq x$ uniformly among all $k-1$
alternatives.

\begin{prop}\label{PropColorSDPI}
	$$\eta_{\KL}(\Col_k, \pi) = \frac{\log k - \log (k-1)}{\log k}.$$
\end{prop}
\begin{proof}
	By~\eqref{eq:coro_sdpi_potts} we have
	\begin{align*}
		\eta_{\KL}(\Col_k, \pi) &= \sup_{x\in (\frac 1k, 1]} \frac{\log k + \frac{1-x}{k-1} \log \frac{1-x}{k-1} + \frac{k+x-2}{k-1}\log \frac{k+x-2}{(k-1)^2}}{\log k + x\log x + (1-x)\log \frac{1-x}{k-1}}\\
		& = \sup_{x\in (\frac 1k, 1]} \frac{\log k - \log (k-1) + \frac{1-x}{k-1} \log (1-x) + \frac{k+x-2}{k-1}\log \frac{k+x-2}{k-1}}{\log k + x\log x + (1-x)\log \frac{1-x}{k-1}}.
	\end{align*}
	Taking $x=1$, we get $$\eta_{\KL}(\Col_k, \pi) \ge \frac{\log k - \log (k-1)}{\log k}.$$
	To prove the proposition, we only need to prove that for $x\in (\frac 1k, 1]$,
	\begin{align}\label{EqnColorSDPI1}
		\frac{\frac{1-x}{k-1} \log (1-x) + \frac{k+x-2}{k-1}\log \frac{k+x-2}{k-1}}{x\log x + (1-x)\log \frac{1-x}{k-1}}
		\ge \frac{\log k-\log(k-1)}{\log k}.
	\end{align}
	(Note that both numerator and denominator in LHS are non-positive.)
	Define
	\begin{align*}
		g(x) &= (\log k - \log(k-1)) x\log x - \frac{\log k}{k-1} (1-x) \log (1-x),\\
	  h(x) &= g(x) + (k-1) g(\frac{1-x}{k-1}).
	\end{align*}
	Rearranging \eqref{EqnColorSDPI1}, we only need to prove that $h(x) \ge 0$ for $x\in (\frac 1k, 1]$.

	We compute that
	\begin{align*}
		g^\p(x) &= (\log k - \log(k-1)) (1+\log x) + \frac{\log k}{k-1} (1+\log (1-x)),\\
		g^\pp(x) &= (\log k - \log(k-1)) \frac 1{x} - \frac{\log k}{k-1} \frac 1 {1-x},\\
		g^\ppp(x) &= -(\log k - \log(k-1)) \frac 1{x^2} - \frac{\log k}{k-1} \frac 1{(1-x)^2}<0.
	\end{align*}

	\begin{claim}\label{ClaimhpppNeg}
		$h^\ppp(x) < 0$ on $(0, 1)$.
	\end{claim}
	\begin{proof}
		\begin{align*}
			h^\ppp(x) &= g^\ppp(x) - \frac 1{(k-1)^2} g^\ppp(\frac{1-x}{k-1})\\
			& = -(\log k - \log(k-1)) \frac 1{x^2} - \frac{\log k}{k-1} \frac 1{(1-x)^2} \\
			& + \frac 1{(k-1)^2} ((\log k - \log(k-1)) \frac 1{(\frac{1-x}{k-1})^2} + \frac{\log k}{k-1} \frac 1{(1-\frac{1-x}{k-1})^2})\\
			& = (\log \frac{k}{k-1}) (\frac 1{(1-x)^2} - \frac 1{x^2}) + \frac {\log k}{k-1} (\frac 1{(k-2+x)^2} - \frac 1{(1-x)^2}) \\
			& = \frac 1{(1-x)^2} ((\log \frac{k}{k-1}) (1 - \frac {(1-x)^2}{x^2}) + \frac {\log k}{k-1} (\frac {(1-x)^2}{(k-2+x)^2}-1))\\
			& =: \frac 1{(1-x)^2} (s(x) + t(x)).
		\end{align*}
		We have
		\begin{enumerate}
			\item $s(x) < 0$ for $x<\frac 12$, $s(x) > 0$ for $x> \frac 12$;
			\item $t(x) < 0$ for $x\in (0, 1)$;
			\item $s(x)$ is increasing for $x\in (0, 1)$;
			\item $t(x)$ is decreasing for $x\in (0, 1)$.
		\end{enumerate}
		So $h^\ppp(x) < 0$ for $x\le \frac 12$.
		For $x \ge \frac 12$, we have
		\begin{align*}
			s(x) + t(x) < s(1) + t(\frac 12) = \log \frac{k}{k-1} + \frac{\log k}{k-1} (\frac 1{(2k-3)^2}-1).
		\end{align*}
		It is not hard to verify that the last term is $<0$ for $k\ge 3$.
	\end{proof}

	By Claim \ref{ClaimhpppNeg}, $h^\p(x)$ is strictly concave.
	Because $h^\p(\frac 1k) = 0$, $h(\frac 1k) = h(1) = 0$, we get that $h(x) > 0$ for $x\in (1/k, 1)$.
	This finishes the proof.
\end{proof}
Theorem \ref{ThmPotts} and Proposition \ref{PropColorSDPI} together imply non-reconstruction for $$\br(T) < \frac{\log k}{\log k - \log(k-1)} = (1-o(1)) k\log k.$$
This result was proved for regular trees by Sly~\cite{Sly09b} and by Bhatnagar et al.~\cite{BVVW11}. Sly had more accurate lower order terms, and his proof works for Galton-Watson trees with Poisson offspring distribution.
Efthymiou \cite{Eft15} generalized the result to general Galton-Watson trees with weak assumptions on the offspring distribution.
Our result does not assume any conditions on the degree distribution other than the expected offspring.

\begin{rmk}\label{RmkColorComparison}
	We show that previous methods based on information contraction do not give the threshold $(1-o(1)) k\log k$.
	The Evans-Schulman method is based on $\eta_{\KL}(\Col_k)$.
	For $P_\ep = (\frac 12 - \ep, \frac 12 + \ep, 0, \ldots, 0)$ and $Q_\ep = (\frac 12 + \ep, \frac 12 - \ep, 0, \ldots, 0)$, we can compute that
	\begin{align*}
		\eta_{\KL}(\Col_k) \ge \lim_{\ep \to 0} \frac{D(\Col_k \circ P_\ep || \Col_k \circ Q_\ep)}{D(P_\ep || Q_\ep)} = \frac 1{k-1}.
	\end{align*}
	On the other hand, by comparing with $\TV$ contraction coefficient, we have
	$$\eta_{\KL}(\Col_k) \le \eta_{\TV}(\Col_k) = \frac 1{k-1}.$$
	Therefore $\eta_{\KL}(\Col_k) = \frac 1{k-1}$. The Evans-Schulman method gives non-reconstruction for $d < k-1$.

	The Formentin-K{\"u}lske method is based on $\eta_{\SKL}(\Col_k, \pi)$.
	If we let $P_\ep = (1-\ep, \frac {\ep}{k-1}, \ldots, \frac{\ep}{k-1})$, then
	\begin{align*}
		\eta_{\SKL}(\Col_k, \pi) \ge \lim_{\ep \to 0} \frac{D_{\SKL}(\Col_k \circ P_\ep || \pi)}{D_{\SKL}(P_\ep || \pi)} = \frac 1{k-1}.
	\end{align*}
	Therefore the Formentin-K{\"u}lske method cannot give non-reconstruction results better than for $d < k-1$.
\end{rmk}

\section{Stochastic block model} \label{SecSBM}
In this section we study the problem of weak recovery of the stochastic block model. In this model, there are $n$ vertices, each independently and uniformly randomly assigned one of $k$ spins.
Say vertex $v\in [n]$ has spin $\sm_v \in [k]$. Generate a random graph,
where two vertices $u$ and $v$ have an edge between them with probability
$$\left\{\begin{array}{ll} \frac an & \text{if }\sm_u = \sm_v,\\ \frac bn & \text{if } \sm_u\ne \sm_v,\end{array}\right.$$
for some absolute constants $a$ and $b$.
We say the model has weak recovery, if given the graph (without knowing the spins $\sm$), we can construct an assignment $\hat \sm$ of spins of the vertices, such that
$$\limsup_{n\to \infty} \bE[d(\sm, \hat \sm)] < 1 - \frac 1k,$$
where $$d(\sm, \hat \sm) = \frac 1n \min_{\tau \in S_k} \sum_{i\in [n]} \bbl\{\sm_i\ne \hat \sm_{\tau(i)}\}.$$
This definition of distance $d$ is meaningful because $\hat \sm$ is in fact defined only up to a permutation of the spins.
If $\hat \sm$ is uniformly random, then the limit of $d(\sm, \hat \sm)$ is $1-\frac 1k$. So the notion of weak recovery indicates whether we can recover the spin groups better than purely random guessing.

In the following, we show that SDPI-based non-reconstruction results for the Potts model lead to improved impossibility results for the stochastic block model.

\subsection{Impossibility of weak recovery via information percolation}
We first give an impossibility result via the information percolation method of Polyanskiy and Wu \cite{PW18}. In op.~cit., they proved the following statement.
\begin{prop}[{\cite[Proposition 8]{PW18}}]\label{PropPW18TreeModel}
	Weak recovery for the stochastic block model is impossible, if the following tree model has non-reconstruction:

	Consider a Galton-Watson tree with offspring distribution Poisson with mean $d = (\sqrt a - \sqrt b)^2$. For each vertex, we independently and uniformly randomly choose a spin.
	Say vertex $v$ has spin $\sm_v$.
	We observe $Y_{u,v} = \bbl\{\sm_u = \sm_v\}$ for each edge $(u,v)$.

	Let $\rho$ denote the root and $L_h$ denote the set of vertices of distance $h$ to $\rho$.
	Let $Y$ denote the set of all observations.
	We say the model has reconstruction, if
	$$ \lim_{h\to \infty} \bE I(\sm_\rho; \sm_{L_h} | Y) = 0.$$
\end{prop}
In op.~cit., it was proven that the tree model has non-reconstruction when $d < \frac k2$, using a coupling argument.
Here we improve it to $$d < \frac {1}{\frac{\log k-\log (k-1)}{\log k} \frac {k-1}k + \frac 1k } = k - (1+o(1)) k/\log k.$$
\begin{prop}\label{PropPW18TreeModel_2}
	The tree model in Proposition \ref{PropPW18TreeModel} has non-reconstruction if
	$$d < \frac {1}{\frac{\log k-\log (k-1)}{\log k} \frac {k-1}k + \frac 1k }.$$
\end{prop}
\begin{proof}
	The tree model is equivalent to the following top-down process:
	\begin{enumerate}
		\item Choose $\sm_\rho$ uniformly randomly over $[k]$.
		\item For an edge $(u,v)$ where $u$ is the parent of $v$, we randomly choose the transition matrix $M$, which is the identity $I_k$ with probability $\frac 1k$, and $\Col_k$ with probability $1-\frac 1k$. Then generate the spin of $v$ according to $\bP(\sm_v=j| \sm_u=i) = M_{i,j}$.
	\end{enumerate}

	For any vertex $u$, define $L_{u,h}$ to be the set of descendants of $u$ that have distance $h$ to $\rho$.
	Let $v$ be a child of $v$.

	Note that $q^* = \Unif([k])$ is an invariant distribution for both $I_k$ and $\Col_k$.
	We have $\eta_{\KL}(I_k^\vee, q^*) = 1$, and by Proposition \ref{PropColorSDPI}, $\eta_{\KL}(\Col_k^\vee, q^*) = \frac{\log k - \log(k-1)}{\log k}$.
	So if $Y_{u,v}=1$, we have
	\begin{align*}
		I(\sm_u; \sm_{L_{v,h}} | Y) \le \eta_{\KL}(I_k^\vee, q^*) I(\sm_v; \sm_{L_{v,h}} | Y)
		= I(\sm_v; \sm_{L_{v,h}} | Y)
	\end{align*}
	and if $Y_{u,v}=0$, we have
	\begin{align*}
		I(\sm_u; \sm_{L_{v,h}} | Y) &\le \eta_{\KL}(\Col_k^\vee, q^*) I(\sm_v; \sm_{L_{v,h}} | Y) \\
		&= \frac{\log k - \log (k-1)}{\log k} I(\sm_v; \sm_{L_{v,h}} | Y).
	\end{align*}
	Taking expectation, we get
	\begin{align*}
		\bE I(\sm_u; \sm_{L_{v,h}} | Y) \le (\frac{\log k - \log (k-1)}{\log k} \frac{k-1}k + \frac 1k) \bE I(\sm_v; \sm_{L_{v,h}} | Y).
	\end{align*}

	Rest of the proof is the same as Theorem \ref{ThmTreeGeneral}.
\end{proof}
The last two propositions together show that the stochastic block model does not have weak recovery when
\begin{equation}\label{EqnSBMInfoPerc}
	(\sqrt a - \sqrt b)^2 < \frac {1}{\frac{\log k-\log (k-1)}{\log k} \frac {k-1}k + \frac 1k}.
\end{equation}
As shown in Figure \ref{FigSBM}, for certain parameters, \eqref{EqnSBMInfoPerc} leads to slight improvement over \cite{BMNN16}.


\subsection{Impossibility of weak recovery via Potts channel}
In the last section we have seen that the information percolation method together with our tree recursion gives a simple yet strong impossibility result for weak recovery of the stochastic block model.
The information percolation method can be understood as comparison with the erasure channel. However, the stochastic block model is more closely related to the Potts channel.
In this section we show an even better impossibility result via the Potts model on a tree.

Let $d = \frac{a + (k-1)b}{k}$ and $\lm = \frac{a-b}{a+(k-1)b}$.
We compare the stochastic block model to the Potts model (with Potts channel parameter $\lm$) on a Galton-Watson tree with offspring distribution $\Po(d)$.

\begin{thm}[Mossel et al.~\cite{MNS15}]\label{ThmSBMPotts}
	If the Potts model with parameter $\lm$ on a Galton-Watson tree with offspring distribution $\Po(d)$ has non-reconstruction, the corresponding stochastic block model does not have weak recovery.
\end{thm}
Mossel et al.~\cite{MNS15} proved Theorem \ref{ThmSBMPotts} for $k=2$, and the proof works for general $k$ with little change. Below we give a sketch of the proof.
\begin{proof}[Proof Sketch]
	The proof is in two parts.
	In the first part, we show that for some absolute constant $c>0$, there exists a coupling between the $c\log n$-neighborhood of a vertex in the SBM, and the $c\log n$-neighborhood of the root in the Potts model on a Galton-Watson tree with offspring distribution $\Po(d)$, such that the total variation distancebetween the two neighborhood (containing spins) is $o(1)$.
	Proof of this part is by observing that the $c\log n$-neighborhood in SBM has no cycle with high probability, and can be constructed using a sequence of binomial random variables; on the other hand, the Potts model can be constructed using a sequence of Poisson variables. Then we compare the two sequences of random variables and find that they have very small total variation distance.

	In the second part, we show that in the SBM, conditioned on the spins on the boundary of the $c\log n$-neighborhood, the spins inside and the spins outside are approximately independent.
	More specifically, if $A, B, C$ is a partition of $V$ such that $\# (A\cup B) = o(n)$ and $B$ separates $A$ and $C$, then
	$$\bP(\sm_A | \sm_{B\cup C}, G) = (1+o(1)) \bP(\sm_A | \sm_B, G)$$
	for $G$ and $\sm$ with probability $1-o(1)$.
	The proof is by writing out the partition function and removing exponents on $1-\frac an$ and $1-\frac bn$ that have negligible effect.

	Combining the two parts, one can prove that for any constant $m$ and any vertices $v_0, \ldots, v_m$, $$I(\sm_{v_0} ; \sm_{v_1}, \ldots, \sm_{v_m} | G) = o(1).$$
	This implies impossibility of weak recovery.
\end{proof}

Theorem \ref{ThmSBMPotts} and Theorem \ref{ThmPotts} imply Theorem \ref{ThmSBM} immediately.
Figure \ref{FigSBM} shows a comparison between the impossibility results for $k=5$.
\begin{figure}[ht]
	\centering
\includegraphics[scale=0.5]{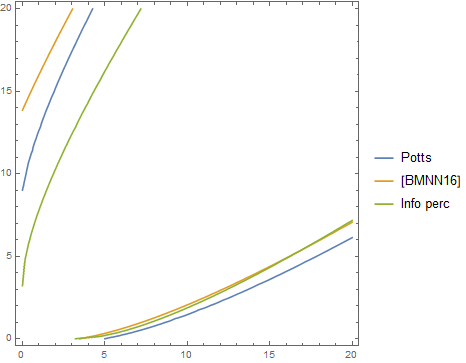}
\caption{Impossibility of weak recovery results for SBM for $k=5$. Horizontal axis is $a$, and vertical axis is $b$. In the assortative regime, \eqref{EqnSBMInfoPerc} gives better results than \cite{BMNN16} for certain parameters, and Theorem \ref{ThmSBM} gives the best results among the three.}
\label{FigSBM}
\end{figure}

Note that \eqref{EqnSBMBMNN} is equivalent to $d \frac{\lm^2(k-1)}{2\log(k-1)} < 1$.
In Appendix \ref{SecUpbd}, we prove that
\begin{align}
	\eta_{\KL}(\PC_\lm, \pi) < \frac{\lm^2(k-1)}{2 \log(k-1)} \label{eq:eta_vs_bmnn}
\end{align}
holds for all $k\ge 3$ and $\lm \in (0, 1]$.
Therefore Theorem \ref{ThmSBM} strictly improves over \eqref{EqnSBMBMNN} in the assortative regime.

\appendix\normalsize

\section{Input-unrestricted contraction coefficient for Potts channels}\label{SecPottsEtaUnres}
Computation of (input-restricted or input-unrestricted) contraction coefficients is often a daunting task.
Previously, Makur and Polyanskiy \cite{MP18} obtained lower and upper bounds of input-unrestricted KL divergence contraction coefficients for Potts channels.
In this appendix we compute the exact value of these contraction coefficients.
\begin{prop}\label{PropConPotts}
	$$\eta_{\KL} (\PC_\lm) = \frac{k\lm^2}{(k-2)\lm+2}.$$
\end{prop}
\begin{proof}
	The result is obvious for $\lm \in \{0, 1\}$. In the following, assume that $\lm\not \in \{0, 1\}$.

	We use the following characterization of contraction coefficient using R\'enyi correlation \cite{Ren59} (see e.g.~Sarmanov~\cite{Sar58}).
	For any channel $M$, we have
	\begin{align}
		\eta_{\KL}(M) = (\sup_P \sup_{f, g} \bE[f(X) g(Y)])^2 \label{EqnEta}
	\end{align}
	where $P$ is a distribution on $[k]$, $X\sim P$, $Y\sim M\circ P$, $f: \cX\to \bR$ satisfies $\bE_X[f] = 0$ and $\bE_X[f^2]=1$,
	and $g: \cY\to \bR$ satisfies $\bE_Y[g]=0$ and $\bE_Y[g^2]=1$.

	Specialize to $M = \PC_\lm$. Write $P = (p_1, \ldots, p_k)$, $f = (f_1, \ldots, f_k)$ and $g = (g_1, \ldots, g_k)$.
	Then
	\begin{align}
		\bE[f(X) g(Y)] = \sum_{i,j} f_i p_i g_j \bP[Y=j|X=i] = \lm \sum f_i p_i g_i. \label{EqnCorr}
	\end{align}
	When $\lm > 0$, we need to maximize $\sum f_i g_i p_i$.
	When $\lm < 0$, we make the transform $f_i \leftarrow -f_i$, and still maximize $\sum f_i g_i p_i$.
	So we get the following optimization problem.
	\begin{align}
		&\max \sum f_i g_i p_i\nonumber \\
		\text{s.t.} \quad & \sum f_i p_i = 0, \label{Eqnf1}\\
		& \sum f_i^2 p_i = 1, \label{Eqnf2}\\
		& \sum g_i (\lm p_i + \frac {1-\lm}{k}) = 0,\label{Eqng1}\\
		& \sum g_i^2 (\lm p_i + \frac {1-\lm}{k}) = 1, \label{Eqng2}\\
		& p_i \ge 0, \sum p_i = 1.
	\end{align}

	\textbf{Lower bound.}
	Take
	\begin{align*}
		P &= (\frac 12, \frac 12, 0, \ldots, 0),\\
		f &= (1, -1, 0, \ldots, 0),\\
		g &= (u, -u, 0, \ldots, 0)
	\end{align*}
	where $$u = \sqrt{\frac{k}{(k-2)\lm+2}}.$$
	Then $$\sum f_i g_i p_i = u.$$
	So $$\eta_{\KL}(\PC_\lm) \ge (\lm u)^2 = \frac{k\lm^2}{(k-2)\lm+2}.$$

	\textbf{Upper bound.}
	Let us fix $P$ and maximize over $f$ and $g$.
	Assume for the sake of contrary that $\sum f_i g_i p_i > u$.
	The set of possible $g$ is bounded; some coordinates of $f$ may be unbounded, but their values do not affect the objective function.
	So the maximum value of $\sum f_i g_i p_i$ is achieved at some point $f$ and $g$.
	Let us compute the derivatives.
	\begin{align*}
		\Na_f \sum f_i g_i p_i &= (g_i p_i)_{i\in [k]},\\
		\Na_f \sum f_i p_i &= (p_i)_{i\in [k]},\\
		\Na_f \sum f_i^2 p_i &= (2 f_i p_i)_{i\in [k]},\\
		\Na_g \sum f_i g_i p_i &= (f_i p_i)_{i\in [k]},\\
		\Na_g \sum g_i (\lm p_i + \frac {1-\lm}{k}) &= (\lm p_i + \frac{1-\lm}k)_{i\in [k]},\\
		\Na_g \sum g_i^2 (\lm p_i + \frac {1-\lm}{k}) &= (2 g_i(\lm p_i + \frac{1-\lm}k))_{i\in [k]}.
	\end{align*}
	By maximality in $f$, there exists some constants $A$ and $B$ such that
	\begin{align}
		g_i p_i = A p_i + B f_i p_i \label{EqnAB}
	\end{align} for all $i$.
	By maximality in $g$, there exists some constants $C$ and $D$ such that
	\begin{align}
		f_i p_i = C (\lm p_i + \frac{1-\lm}k) + D g_i (\lm p_i + \frac{1-\lm}k) \label{EqnCD}
	\end{align}
	for all $i$.

	By \eqref{EqnAB},
	\begin{align}
		\sum f_i g_i p_i = \sum f_i (A p_i + B f_i p_i) = B.
	\end{align}
	By \eqref{EqnCD},
	\begin{align}
		\sum f_i g_i p_i = \sum g_i (C (\lm p_i + \frac{1-\lm}k) + D g_i (\lm p_i + \frac{1-\lm}k)) = D.
	\end{align}
	So $B=D > u > 0$.

	For $p_i\ne 0$, we have
	$g_i = A + B f_i$
	by \eqref{EqnAB}.

	If for some $i$, $p_i=0$, then
	$$\frac {1-\lm} k(C+D g_i) = 0.$$
	This means $$\#\{g_i : p_i=0\} = 1.$$
	So we can choose $f_i$ for such $i$ such that
	\begin{align}
		g_i = A + B f_i \label{EqnABM}
	\end{align}
	for all $i$.

	From \eqref{Eqng1}, we get
	\begin{align*}
		0 & = \sum g_i (\lm p_i + \frac {1-\lm}k) \\
		& = \sum (A + B f_i) (\lm p_i + \frac {1-\lm}k)\\
		& = A + B \frac {1-\lm}k \sum f_i.
	\end{align*}
	From \eqref{Eqng2}, we get
	\begin{align*}
		1 &= \sum g_i^2 (\lm p_i + \frac {1-\lm}k) \\
		& = \sum (A^2 + 2AB f_i + B^2 f_i^2) (\lm p_i + \frac {1-\lm}k) \\
		& = A^2 + 2AB \frac{1-\lm} k \sum f_i + B^2 \lm + B^2 \frac{1-\lm} k \sum f_i^2 \\
		& = B^2 ( \lm + \frac {1-\lm}k \sum f_i^2 - (\frac {1-\lm}k \sum f_i)^2 ).
	\end{align*}

	The result then follows from Claim \ref{ClaimIneq} because we have
	\begin{align*}
		B & = \frac 1{\sqrt{\lm + \frac {1-\lm} k (\sum f_i^2 - \frac {1-\lm}k (\sum f_i)^2)}} \\
		& \le \frac 1{\sqrt{\lm + \frac {1-\lm} k (\sum f_i^2 - \frac 1{k-1} (\sum f_i)^2)}} \\
		& \le \frac 1{\sqrt{\lm + \frac {1-\lm} k \cdot 2}} = u.
	\end{align*}
\end{proof}

\begin{claim}\label{ClaimIneq}
	For any distribution $P$ and any $f$ satisfying \eqref{Eqnf1} and \eqref{Eqnf2},
	we have
	\begin{align*}
		\sum f_i^2 - \frac 1{k-1} (\sum f_i)^2 \ge 2.
	\end{align*}
\end{claim}
\begin{proof}
	Let us first prove the result for $f$ with support size two.
	WLOG assume that $f_1>0$, $f_2<0$, $f_3 = \cdots = f_k=0$.
	One can compute that
	\begin{align*}
		f_1 = \sqrt{\frac{p_2}{p_1(p_1+p_2)}}, \quad
		f_2 = -\sqrt{\frac{p_1}{p_1(p_1+p_2)}}.
	\end{align*}
	Then
	\begin{align*}
		&f_1^2+f_2^2-\frac 1{k-1} (f_1+f_2)^2 \\
		& \ge f_1^2+f_2^2-(f_1+f_2)^2\\
		&= \frac 1{p_1+p_2} (\frac{p_2}{p_1}+\frac{p_1}{p_2} - (\sqrt{\frac{p_2}{p_1}}-\sqrt{\frac{p_1}{p_2}})^2) \\
		&= \frac 2{p_1+p_2} \ge 2.
	\end{align*}

	Let us define
	\begin{align*}
		S(P) &:= \{f : \sum f_ip_i=0, \sum f_i^2 p_i = 1\}\\
		U(f) &:= \sum f_i^2 - \frac 1{k-1} (\sum f_i)^2.
	\end{align*}

	Now suppose that for some $P$ and $f \in S(P)$ we have $U(f) < 2$.
	The set $S(P) / \{\pm\}$ is continuous,
	and there exists $f\in S(P)$ with $U(f) \ge 2$ (e.g., $f$ with support size two),
	so for sufficiently small $\ep>0$ there exists $f\in S(P)$ such that $U(f) \in (2-\ep, 2)$.

	Let $\lm = -\frac 1{k-1}$.
	Take $\ep$ small enough so that $\lm + \frac {1-\lm} k (2-\ep) > 0$ and choose $f\in S(P)$ with $U(f) \in (2-\ep, 2)$.
	Define
	\begin{align*}
		B &= \frac 1{\sqrt{\lm + \frac {1-\lm} k U(f)}} > u,\\
		A &= - B \frac {1-\lm}k \sum f_i,\\
		g_i &= A + B f_i \forall i.
	\end{align*}
	One can check that $g$ satisfies \eqref{Eqng1} and \eqref{Eqng2},
	and
	\begin{align*}
		\sum f_i g_i p_i = B > u.
	\end{align*}

	By \eqref{EqnEta} and \eqref{EqnCorr}, this implies
	\begin{align*}
		\eta_{\KL}(\PC_{-\frac 1{k-1}}) > \frac 1{k-1}.
	\end{align*}
	However, we have
	\begin{align*}
		\eta_{\KL}(\PC_{-\frac 1{k-1}}) \le \eta_{\TV} (\PC_{-\frac 1{k-1}}) = \frac 1{k-1}.
	\end{align*}
	Contradiction.
\end{proof}

\section{An upper bound for input-restricted contraction coefficient for Potts channels}\label{SecUpbd}
In this appendix we prove an upper bound for the input-restricted KL divergence contraction coefficient for ferromagnetic Potts channels.
\begin{prop}\label{PropUpbd}
	Fix $k\ge 3$.
	For all $\lm \in [0, 1]$, we have
	\begin{equation}\label{eq:eta_upbd_ferro}
		\eta_{\KL}(\PC_\lm, \pi) \le \frac{\lm^2}{(1-\lm) \frac{2(k-1)\log(k-1)}{k(k-2)}+ \lm}.
	\end{equation}
	For all $\lm \in [-\frac 1{k-1}, 0]$, we have
	\begin{equation}\label{eq:eta_upbd_antiferro}
		\eta_{\KL}(\PC_\lm, \pi) \le \frac{\lm^2}{(1+ (k-1)\lm)\frac{2(k-1)\log(k-1)}{k(k-2)} - \lm \frac{\log k}{(k-1)(\log k - \log (k-1))}}.
	\end{equation}
\end{prop}
We first prove a lemma.
\begin{lemma}\label{LemmaConcTech}
	$\frac{(kx-1)^2}{\psi(x)}$ is concave in $x\in [0, 1]$.
\end{lemma}
\begin{proof}
	Let $f(x) = \frac{(kx-1)^2}{\psi(x)}$.
	\begin{align*}
		f^\p(x) &= \frac{2k (kx-1)}{\psi(x)} - \frac{(kx-1)^2 \psi^\p(x)}{\psi^2(x)}.\\
		f^\pp(x) &= \frac{2k^2}{\psi(x)} - \frac{4k(kx-1) \psi^\p(x)}{\psi^2(x)} -\frac{(kx-1)^2\psi^\pp(x)}{\psi^2(x)} + \frac{2 (kx-1)^2\psi^{\p 2} (x)}{\psi^3(x)}\\
		& = \frac 2{\psi^3(x)} (k\psi(x)-(kx-1)\psi^\p(x))^2 - \frac{(kx-1)^2\psi^\pp(x)}{\psi^2(x)}.
	\end{align*}
	Therefore it suffices to prove that
	\begin{align*}
		g(x) := \psi^3(x) f^\pp(x) = 2(k\psi(x)-(kx-1)\psi^\p(x))^2 - (kx-1)^2 \psi(x) \psi^\pp(x)
	\end{align*}
	is non-positive for $x\in [0, 1]$.
	Note that $g(\frac 1k)=0$. So we only need to prove that
	$g^\p(x)\ge 0$ for $x\in [0, \frac 1k]$ and $g^\p(x) \le 0$ for $x\in [\frac 1k, 1]$.
	\begin{align*}
		g^\p(x) &= -4(kx-1)\psi^\pp(x) (k\psi(x) - (kx-1) \psi^\p(x)) - 2k (kx-1) \psi(x) \psi^\pp(x) \\
		&- (kx-1)^2 \psi^\p(x) \psi^\pp(x) - (kx-1)^2 \psi(x)\psi^\ppp(x)\\
		& = (kx-1) (-6k \psi(x) \psi^\pp(x) + (kx-1) (3 \psi^\p(x)\psi^\pp(x) - \psi(x) \psi^\ppp(x))).
	\end{align*}
	Therefore we would liek to prove that
	$$u(k,x) := -6k \psi(x) \psi^\pp(x) + (kx-1) (3 \psi^\p(x)\psi^\pp(x) - \psi(x) \psi^\ppp(x))$$
	is non-positive.
	We enlarge the domain of $u$ and prove that $u(k,x)\le 0$ for real $k>1$ and $x\in (0, 1)$.

	We fix $x\in (0, 1)$ and consider $u_x(k) := u(k, x)$.
	We have $u_x(\frac 1x)=0$. So it suffices to prove that $u_x$ is concave in $k$.
	We have
	\begin{align*}
		\psi^\p(x) &= \log x - \log \frac{1-x}{k-1},\\
		\psi^\pp(x) &= \frac 1x + \frac 1{1-x},\\
		\psi^\ppp(x) &= \frac 1{(1-x)^2} - \frac 1{x^2},\\
		\frac{\der}{\der k} \psi(x) &= \frac 1k - \frac{1-x}{k-1},\\
		\frac{\der}{\der k} \psi^\p(x) &= \frac 1{k-1},\\
		\frac{\der}{\der k} \psi^\pp(x) &= \frac{\der}{\der k} \psi^\ppp(x)=0.
	\end{align*}
	So
	\begin{align*}
		u_x^\p(k) &= -6 \psi(x) \psi^\pp(x) - 6 k (\frac 1k - \frac{1-x}{k-1}) \psi^\pp(x)
		+ x (3 \psi^\p(x) \psi^\pp(x) - \psi(x) \psi^\ppp(x))\\
		&+ (kx-1) (3 \frac 1{k-1} \psi^\pp(x) - (\frac 1k - \frac {1-x}{k-1})\psi^\ppp(x)).\\
		u_x^\pp(k) &= -12 (\frac 1k - \frac{1-x}{k-1}) \psi^\pp(x)
		-6k (-\frac 1{k^2} + \frac{1-x}{(k-1)^2}) \psi^\pp(x)\\
		&+6x \frac 1{k-1} \psi^\pp(x) - 2x (\frac 1k-\frac{1-x}{k-1}) \psi^\ppp(x)\\
		&+(kx-1)(-3 \frac1{(k-1)^2}\psi^\pp(x) -(-\frac 1{k^2} + \frac{1-x}{(k-1)^2})\psi^\ppp(x))\\
		& = \frac{(kx-1)^2(1-2k+(k-2)x)}{k^2(k-1)^2 x^2(1-x)^2}\le 0.
	\end{align*}
	We are done.
\end{proof}
\begin{proof}[Proof of Proposition \ref{PropUpbd}]
	For fixed $x$, we would like to lower bound
	\begin{align*}
		f_x(\lm) := \frac{\lm^2 \psi(x)}{\psi(\lm x + \frac{1-\lm}k)}.
	\end{align*}
	(Value of $f_x(0)$ is defined by continuity.)
	Because
	\begin{align*}
		f_x(\lm) = \frac{\psi(x)}{(kx-1)^2} \cdot \frac{(k(\lm x + \frac{1-\lm}k)-1)^2}{\psi(\lm x + \frac{1-\lm}k)},
	\end{align*}
	by Lemma \ref{LemmaConcTech},
	$f_x(\lm)$ is concave for $\lm \in [-\frac 1{k-1}, 1]$.

	Let us compute lower bounds of $f_x(\lm)$ for $\lm = -\frac 1{k-1}, 0, 1$.

	By Proposition \ref{PropColorSDPI}, we have
	\begin{equation}\label{eq:bd_at_antiferro}
		f_x(-\frac 1{k-1}) \ge \frac {\log k}{(k-1)^2 (\log k - \log (k-1))}.
	\end{equation}

	By L'H\^{o}pital's rule,
	\begin{align*}
		f_x(0) &= \psi(x) \lim_{\lm\to 0} \frac{2\lm}{(x-\frac 1k)\psi^\p(\lm x + \frac{1-\lm}k)}\\
		& = \psi(x) \lim_{\lm\to 0} \frac 2{(x-\frac 1k)^2\psi^\pp(\lm x + \frac{1-\lm}k)}\\
		& = \frac{2(k-1) \psi(x)}{(kx-1)^2}.
	\end{align*}
	By Lemma \ref{LemmaConcTech}, $g(x):= \frac{(kx-1)^2}{\psi(x)}$ is concave in $x$.
	Also
	\begin{align*}
		g^\p(1-\frac 1k) &= \frac{2k(k-2)}{\frac 1k (k-2) \log(k-1)}
		-\frac{(k-2)^2 \cdot 2\log (k-1)}{(\frac 1k (k-2) \log(k-1))^2}
		=0.
	\end{align*}
	So
	\begin{align*}
		g(x) \le g(1-\frac 1k) = \frac{k(k-2)}{\log (k-1)}
	\end{align*}
	and
	\begin{equation}\label{eq:bd_at_0}
		f_x(0) \ge \frac{2(k-1) \log(k-1)}{k(k-2)}.
	\end{equation}

	It is easy to see that
	\begin{equation}\label{eq:bd_at_1}
		f_x(1) \ge 1.
	\end{equation}

	Because $f_x(\lm)$ is concave in $\lm$, Inequality \eqref{eq:eta_upbd_ferro} follows from \eqref{eq:bd_at_0} and \eqref{eq:bd_at_1},
	and Inequality \eqref{eq:eta_upbd_antiferro} follows from \eqref{eq:bd_at_antiferro} and \eqref{eq:bd_at_0}.
\end{proof}
\begin{rmk}
	Proof of Proposition \ref{PropUpbd} implies the first order limit behavior of $\eta_{\KL}(\PC_\lm, \pi)$ as $\lm \to 0$.
	\begin{align*}
		\lim_{\lm \to 0} \frac{\eta_{\KL}(\PC_\lm, \pi)}{\lm^2} = \frac{k(k-2)}{2(k-1)\log(k-1)}.
	\end{align*}

	Note that for all $k\ge 3$ and $\lm\in(0, 1]$,
	\begin{align*}
		\eta_{\KL}(\PC_\lm, \pi) &\le \frac{\lm^2}{(1-\lm) \frac{2(k-1)\log(k-1)}{k(k-2)}+ \lm}\\
		&\le  \lm^2(1-\lm)\frac{k(k-2)}{2(k-1)\log(k-1)}+\lm^3\\
		& < \lm^2 \frac{k(k-2)}{2(k-1)\log(k-1)}\\
		& < \lm^2 \frac{k-1}{2\log(k-1)}.
	\end{align*}
	(Second step is by Cauchy inequality.)
	So \eqref{eq:eta_upbd_ferro} implies \eqref{eq:eta_vs_bmnn}.

	For comparison with input-unrestricted contraction coefficient $$\eta_{\KL}(\PC_\lm) = \frac{k\lm^2}{(k-2)\lm+2},$$
	we note that
	$\frac{\lm^2}{\eta_{\KL}(\PC_\lm)}$ is linear in $\lm$, and
	\begin{align*}
		\frac 1{k-1} &< \frac{\log k}{(k-1)^2(\log k-\log(k-1))},\\
		\frac 2k &< \frac{2(k-1)\log(k-1)}{k(k-2)}.
	\end{align*}
	So Proposition \ref{PropUpbd} implies \eqref{eq:eta_res_vs_unres}.
\end{rmk}

\section{Non-convexity of certain functions}\label{SecNonConv}
In this section we prove Proposition \ref{PropNonConv}.
Let us first prove a lemma.
\begin{lemma}\label{LemmaConvGen}
	Let $g$ be a strictly increasing smooth function from $[x_0, x_1]$ to $[y_0, y_1]$,
	and $f$ be a smooth function from $[x_0, x_1]$ to $\bR$.
	Assume that $g^\p(x_0)=f^\p(x_0)=0$ and $(g^\pp f^\ppp - f^\pp g^\ppp)(x_0)>0$.
	Then the function $h=f\circ g^{-1}: [y_0, y_1] \to \bR$ is not concave near $y_0$.
\end{lemma}
\begin{proof}
	Directives of $h$ are
	\begin{align*}
		h^\p(x) &= \frac{f^\p(g^{-1}(x))}{g^\p(g^{-1}(x))},\\
		h^\pp(x) &= (\frac{f^\pp}{g^\p}-\frac{f^\p g^\pp}{g^{\p 2}})(g^{-1}(x)) \frac 1{g^\p(g^{-1}(x))} \\
		& =(\frac{f^\pp}{g^{\p 2}} - \frac {f^\p g^\pp} {g^{\p 3}})(g^{-1}(x)).
	\end{align*}
	So it suffices to study the sign of $g^\p f^\pp - f^\p g^\pp$ for $x$ near $x_0$.
	Let $u = g^\p f^\pp - f^\p g^\pp$.
	We have $u(x_0)=0$. Let us compute the derivatives.
	\begin{align*}
		u^\p &= g^\p f^\ppp - f^\p g^\ppp,\\
		u^\pp &= g^\p f^{(4)} + g^\pp f^\ppp - f^\pp g^\ppp - g^\p g^{(4)}.
	\end{align*}
	So $u^\p(x_0)=0$ and $u^\pp(x_0) = (g^\pp f^\ppp - f^\pp g^\ppp)(x_0) > 0$.
	So $u$ is positive near $x_0$.
\end{proof}

\begin{proof}[Proof of Proposition \ref{PropNonConv}]
	We apply Lemma \ref{LemmaConvGen} to $g=\psi$, $x_0=\frac 1k$, $x_1=1$, $y_0=0$, $y_1=\log k$, and various $f$.
	We have
	\begin{align*}
		\psi^\p(\frac 1k) &= 0,\\
		\psi^\pp(\frac 1k) &= \frac {k^2}{k-1},\\
		\psi^\ppp(\frac 1k) &= -\frac{k^3(k-2)}{(k-1)^2}.
	\end{align*}

	\textbf{Part 1.}
	For $b_1$, take
	\begin{align*}
		f(x) = -(k-1)\xi_1(x) = \log x + (k-1)\log \frac{1-x}{k-1} + k(\psi(x)-\log k).
	\end{align*}
	Then
	\begin{align*}
		f^\p(x) &= \frac 1x -\frac{k-1}{1-x} + k\psi^\p(x),\\
		f^\pp(x) &= -\frac 1{x^2} - \frac{k-1}{(1-x)^2} + k \psi^\pp(x),\\
		f^\ppp(x) &= \frac 2{x^3} - \frac{2(k-1)}{(1-x)^3} + k \psi^\ppp(x).
	\end{align*}
	So
	\begin{align*}
		f^\p(\frac 1k)&=0,\\
		f^\pp(\frac 1k)&=-\frac{2k^3}{k-1},\\
		f^\ppp(\frac 1k)&=\frac{3(k-2)k^4}{(k-1)^2}.
	\end{align*}
	We have
	\begin{align*}
		(\psi^\pp f^\ppp - f^\pp \psi^\ppp)(\frac 1k)
		&= \frac {k^2}{k-1}\cdot \frac{3(k-2)k^4}{(k-1)^2}
		-(-\frac{2k^3}{k-1})(-\frac{k^3(k-2)}{(k-1)^2})\\
		&= \frac{k^6(k-2)}{(k-1)^3} > 0.
	\end{align*}
	So Lemma \ref{LemmaConvGen} applies.

	\textbf{Part 2.}
	For $b_p$, $p>1$, take $$f(x) = k-(k-1)\xi_p(x) = (x^{\frac 1p} + (k-1) (\frac {1-x}{k-1})^{\frac 1p})(x^{1-\frac 1p} + (k-1) (\frac {1-x}{k-1})^{1-\frac 1p}).$$
	For simplicity, write $r = \frac 1p$ and let $u_r(x) = x^r + (k-1) (\frac{1-x}{k-1})^r$.
	Then $f(x) = u_r(x) u_{1-r}(x)$.
	Let us compute derivatives of $u_r$.
	\begin{align*}
		u_r^\p(x) &= r (x^{r-1} - (\frac{1-x}{k-1})^{r-1}),\\
		u_r^\pp(x) &= r(r-1) (x^{r-2} + \frac 1{k-1} (\frac{1-x}{k-1})^{r-2}),\\
		u_r^\ppp(x) &= r(r-1) (r-2) (x^{r-3} - \frac 1{(k-1)^2} (\frac{1-x}{k-1})^{r-3}).
	\end{align*}
	So
	\begin{align*}
		u_r(\frac 1k)&=k^{1-r},\\
		u_r^\p(\frac 1k)&=0,\\
		u_r^\pp(\frac 1k)&=r(r-1) \frac k{k-1} (\frac 1k)^{r-2},\\
		u_r^\ppp(\frac 1k)&= r(r-1)(r-2) \frac {k(k-2)}{(k-1)^2} (\frac 1k)^{r-3}.
	\end{align*}
	Now we compute derivatives of $f$.
	\begin{align*}
		f^\p(x) &= u_r^\p(x) u_{1-r}(x) + u_r(x) u_{1-r}^\p(x),\\
		f^\pp(x) &= u_r^\pp(x) u_{1-r}(x) + 2u_r^\p(x) u_{1-r}^\p(x)+ u_r(x) u_{1-r}^\pp(x),\\
		f^\ppp(x) &= u_r^\ppp(x) u_{1-r}(x) + 3u_r^\pp(x) u_{1-r}^\p(x) + 3u_r^\p(x) u_{1-r}^\pp(x) + u_r(x) u_{1-r}^\ppp(x).
	\end{align*}
	So
	\begin{align*}
		f^\p(\frac 1k) &= 0,\\
		f^\pp(\frac 1k) &= r(r-1) \frac k{k-1} (\frac 1k)^{r-2} \cdot k^r + (1-r)(-r) \frac k{k-1} (\frac 1k)^{-r-1} \cdot k^{1-r}\\
		&= 2r(r-1) \frac {k^3}{(k-1)},\\
		f^\ppp(\frac 1k) &= r(r-1)(r-2) \frac {k(k-2)}{(k-1)^2} (\frac 1k)^{r-3} \cdot k^r\\
		&+ (1-r)(-r)(-r-1) \frac {k(k-2)}{(k-1)^2} (\frac 1k)^{-r-2} \cdot k^{1-r}, \\
		&=-3r(r-1) \frac{k^4(k-2)}{(k-1)^2}.
	\end{align*}
	So
	\begin{align*}
		(\psi^\pp f^\ppp - f^\pp \psi^\ppp)(\frac 1k)
		&= \frac {k^2}{k-1}(-3r(r-1) \frac{k^4(k-2)}{(k-1)^2})\\
		& -2r(r-1) \frac {k^3}{(k-1)} (-\frac{k^3(k-2)}{(k-1)^2})\\
		&=r(1-r) \frac{k^6(k-2)}{(k-1)^3} > 0.
	\end{align*}
	So Lemma \ref{LemmaConvGen} applies.

	\textbf{Part 3.}
	For $s_\lm$, take $$f(x) = \psi(\lm x + \frac{1-\lm}k).$$
	Then
	\begin{align*}
		f^\p(x) &= \lm \psi^\p(\lm x + \frac{1-\lm}k),\\
		f^\pp(x) &= \lm^2 \psi^\pp(\lm x + \frac{1-\lm}k),\\
		f^\ppp(x) &= \lm^3 \psi^\ppp(\lm x + \frac{1-\lm}k).
	\end{align*}
	So
	\begin{align*}
		f^\p(\frac 1k)&=0,\\
		f^\pp(\frac 1k) &= \lm^2 \psi^\pp(\frac 1k) = \lm^2\frac {k^2}{k-1},\\
		f^\ppp(\frac 1k) &= \lm^3 \psi^\ppp(\frac 1k) = -\lm^3 \frac{k^3(k-2)}{(k-1)^2}.
	\end{align*}
	We have
	\begin{align*}
		(\psi^\pp f^\ppp - f^\pp \psi^\ppp)(\frac 1k)
		&= \frac {k^2}{k-1}(-\lm^3 \frac{k^3(k-2)}{(k-1)^2}) -
		\lm^2\frac {k^2}{k-1}(-\frac{k^3(k-2)}{(k-1)^2})\\
		&=\frac{k^5(k-2)}{(k-1)^3} (\lm^2-\lm^3) > 0.
	\end{align*}
	So Lemma \ref{LemmaConvGen} applies.
\end{proof}

\section{Concavity of log-Sobolev coefficients}\label{apx:lsi_conc}

Let $K$ be a Markov kernel with stationary distribution $\pi$.
Define Dirichlet form $\cE(\cdot, \cdot)$ and entropy form $\Ent(\cdot)$ as in the introduction.

For $r\in \bR$, we consider the tightest $\frac1r$-log-Sobolev inequality, corresponding to
\begin{align*}
	b_{\frac1r}(x) :&= \inf_{\substack{f: \cX \to \bR_{\ge 0},\\ \bE_\pi f=1, \Ent_\pi(f)=x}} \cE(f^r, f^{1-r}),\\
	\Phi_{\frac 1r}(y) &:= \inf_{\substack{f: \cX \to \bR_{\ge 0},\\ \bE_\pi f=1, \cE(f^r, f^{1-r})=y}} \Ent_\pi(f).
\end{align*}
The $\frac1r$-log-Sobolev constant is
\begin{align*}
	\al_{\frac 1r}^\p := \inf_{x>0} \frac{b_{\frac 1r}(x)}x = \inf_{y>0} \frac{y}{\Phi_{\frac 1r}(y)}.
\end{align*}
\begin{rmk}
	When $r=0$, the fraction $\frac 1r$ should be understood as a formal symbol.
	For $r\in (0, 1)$, $\al_{\frac 1r}^\p$ is the same as $\al_{\frac 1r}$ defined in the Introduction.
	However, in general $\al_1^\p$ is not equal to $\al_1$.
	We use the superscript $\p$ to emphasize the difference.
\end{rmk}
\begin{prop}\label{PropLSIConc}
	We have
	\begin{enumerate}
		\item For fixed $x$, $b_{\frac 1r}(x)$ is concave in $r$.
		\item For fixed $y$, $\Phi_{\frac 1r}(y)$ is convex in $r$.
		\item $\al_{\frac 1r}^\p$ is concave in $r$.
	\end{enumerate}
	Furthermore, if $(K, \pi)$ is reversible, then
	\begin{enumerate}
		\item For fixed $x$, $b_{\frac 1r}(x)$ is maximized at $r=\frac 12$.
		\item For fixed $y$, $\Phi_{\frac 1r}(y)$ is minimized at $r=\frac 12$.
		\item $\al_{\frac 1r}^\p$ is maximized at $r=\frac 12$.
	\end{enumerate}
\end{prop}
\begin{proof}
	Because $\Phi_{\frac1r}$ is the inverse function of $b_{\frac 1r}$, it suffices to prove
	statements about $b_{\frac 1r}$.
	Because $\inf$ of concave functions is still concave, it suffices to prove that
	for any $f: \cX\to \bR_{\ge 0}$, $\bE_\pi f=1$, $\cE(f^r, f^{1-r})$ is concave in $r$.

	\begin{align*}
		\frac{d}{dr^2} \cE(f^r, f^{1-r})
		& = \frac{d}{dr^2} \sum_{x,y\in \cX} (I-K)(x, y) f(y)^r f(x)^{1-r} \pi(x)\\
		& = \sum_{x,y\in \cX} (I-K)(x, y) f(y)^r f(x)^{1-r} \pi(x) (\log f(y)-\log f(x))^2\\
		& = \sum_{x\ne y\in \cX} -K(x, y) f(y)^r f(x)^{1-r} \pi(x) (\log f(y)-\log f(x))^2\\
		& \le 0.
	\end{align*}

	When the Markov chain is reversible, we have $\cE(f, g) = \cE(g, f)$.
	So $b_{\frac 1r}(x) = b_{\frac 1{1-r}}(x)$ and by concavity, $b_{\frac 1r}(x)$ is maximized at $r=\frac 12$.
\end{proof}

\section{Non-reconstruction results for broadcasting with a Gaussian kernel} \label{SecBOTGauss}
In this section, we prove optimal non-reconstruction results for a broadcasting on trees model considered in Eldan et al.~\cite{Eld20}, using our method developed in Section~\ref{SecTreeGeneral}.

\begin{defn}[Broadcasting on trees with a Gaussian kernel]
In this model, we are given a (possibly) infinite tree $T$ with a marked root $\rho$.
The state space $\cX$ is the unit circle $S^1 := \bR/2\pi \bZ$.
Let $\mu = \Unif(S^1)$ be the uniform distribution.
Let $t>0$ be a parameter.
The transfer kernel is $P_t$, defined as $Y = X + Z_t$ where $Z_t \sim \cN(0, t)$, where $X$ is the input and $Y$ is the output.
It is easy to see that $P_t$ is reversible under $\mu$, i.e., $P_t^\vee = P_t$.

Now for each vertex $v\in T$, we generate a spin $\sigma_v \in \cX$ according to the following rules:
\begin{itemize}
	\item Generate $\sigma_\rho \sim \mu$.
	\item If $u$ is the parent of $v$, then $\sigma_v \sim P_t(\cdot | \sigma_u)$.
\end{itemize}

Let $L_h$ denote the set of vertices at distance $h$ to $\rho$. We say the model has non-reconstruction if and only if $\lim_{h\to \infty} I(\sigma_\rho; \sigma_{L_h}) = 0$.
\end{defn}

Let $\lambda(P_t)$ denote the second largest eigenvalue of $P_t$.
Eldan et al.~\cite{Eld20} proved that reconstruction holds for a regular tree with offspring number $b$, reconstruction holds for $b\lambda(P_t)^2 > 1$, and non-reconstruction holds for $b \lambda(P_t) < 1$.
Note that there is a $\lambda(P_t)$ factor gap between the reconstruction result and the non-reconstruction result.
In the following, we prove that non-reconstruction holds as long as $b\lambda(P_t)^2 < 1$, closing the gap.

\begin{thm} \label{ThmBOTGauss}
	Let $T$ be an infinite rooted tree with bounded maximum degree. Then for broadcasting on trees with $(\mu, P_t)$, non-reconstruction holds when
	\begin{align*}
		\br(T) \lambda(P_t)^2 < 1.
	\end{align*}

	Let $T$ be a Galton-Watson tree with expected offspring $b$. If
	\begin{align*}
		b \lambda(P_t)^2 < 1,
	\end{align*}
	then non-reconstruction holds almost surely.
\end{thm}


The proof idea is to prove $\eta_{\KL}(P_t, \mu) \le \lambda(P_t)^2$ and then use Theorem~\ref{ThmTreeGeneral}. However, because we are working in a continuous space, we must be careful about what we mean by contraction coefficients.

We would like an inequality of form
\begin{align*}
	I(\sigma_u; \sigma_{L_{v,h}}) \le \wt \eta_{\KL} (P_t, \mu) I(\sigma_v; \sigma_{L_{v,h}})
\end{align*}
where $u\in V(T)$, $v$ is child of $u$, $L_{v,h}$ is the set of descendants of $v$ at distance $h$ to $\rho$, and $\wt \eta_{\KL} (P_t, \mu)$ is a continuous version of contraction coefficient $\eta_{\KL}(P_t, \mu)$ (we use $\wt \eta_{\KL}$ to emphasize the difference with $\eta_{\KL}(P_t, \mu)$, which we defined only for the discrete case).

We have
\begin{align*}
	I(\sigma_u; \sigma_{L_{v,h}}) &= \bE_{\sigma_{L_{v,h}}} D(P_{\sigma_u | \sigma_{L_{v,h}}} \| P_{\sigma_u}) \\
	&= \bE_{\sigma_{L_{v,h}}} D(P_t \circ P_{\sigma_v | \sigma_{L_{v,h}}} \| \mu).
\end{align*}
Let us consider the distribution $P_{\sigma_u | \sigma_{L_{u,h}}}$.
If $h = d(v,\rho)$, then $P_{\sigma_u | \sigma_{L_{u,h}}}$ is a point measure.
However, as long as $h > d(u,\rho)$, pdf of $P_{\sigma_u | \sigma_{L_{u,h}}}$ is smooth on $\cX$ by an induction using belief propagation equation.
Therefore we make the following definition.
\begin{defn}[Smooth contraction coefficient] \label{DefnSmoothContCoef}
	We define
	\begin{align*}
		\wt \eta_{\KL}(P_t, \mu) := \sup_{\substack{f: \cX \to \bR_{\ge 0}\\ f \text{ smooth}, \mu[f] = 1} }\frac{\Ent_\mu(P_t f)}{\Ent_\mu(f)}.
	\end{align*}
\end{defn}
\begin{lemma} \label{LemmaSmoothContCoef}
	\begin{align*}
		\wt \eta_{\KL}(P_t, \mu) \le \exp(-t).
	\end{align*}
\end{lemma}
\begin{proof}
	For convenience, we define the following collection of functions
	\begin{align*}
		\cC := \{f: \cX \to \bR_{\ge 0} | f\text{ smooth}, \mu[f] = 1\}.
	\end{align*}

	Note that $(P_t)_{t\ge 0}$ forms a semigroup. Therefore it suffices to prove that for all $f\in \cC$, we have
	\begin{align*}
		\frac{d}{dt}|_{t=0} \Ent(f_t) \le -\Ent(f)
	\end{align*}
	where $f_t = P_t f$.

	We have
	\begin{align*}
		&~ \frac{d}{dt}|_{t=0} \Ent(f_t) \\
		=&~ \bE [\frac {d}{dt}|_{t=0} (f_t \log f_t)] \\
		=&~ \bE [(1 + \log f) \frac{d}{dt}|_{t=0} f_t] \\
		=&~ \bE [(\log f) \frac{d}{dt}|_{t=0} f_t]  \\
		=&~ \frac 12 \bE [f'' \log f] & \text{(Heat equation)}\\
		=&~ -\frac 12 \bE [\frac{(f')^2}{f}] & \text{(Integration by parts)} \\
		\le&~ -\Ent(f). & \text{(\cite{Eme87})}
	\end{align*}

	This finishes the proof.
\end{proof}

Now we are ready to prove Theorem~\ref{ThmBOTGauss}.
\begin{proof}[Proof of Theorem~\ref{ThmBOTGauss}]
	By Lemma~\ref{LemmaSmoothContCoef}, we have
	\begin{align*}
		\wt \eta_{\KL}(P_t,u) \le \exp(-t) = \lambda(P_t)^2.
	\end{align*}
	(The value of $\lambda(P_t)$ is proved in e.g.,~\cite{Eld20}.)
	Therefore we only need to prove that $\br(T) \wt \eta_{\KL}(P_t,u) < 1$ implies non-reconstruction.

	\textbf{Bounded degree case:}
	For $u\in V(T)$, define
	\begin{align*}
		r_u := \lim_{h\to \infty} I(\sigma_u; \sigma_{L_{u,h}}).
	\end{align*}
	By data processing inequality, $I(\sigma_{u}; \sigma_{L_{u,h}})$ is non-increasing for $h\ge d(u,\rho)$, so the limit always exists.
	Because $T$ has bounded maximum degree, we have $r_u \le I(\sigma_u; \sigma_{L_{u, d(u,\rho)+1}})$.
	So there exists a constant $C>0$ such that $r_u\le C$ for all $u\in v(T)$.

	Now define
	\begin{align*}
		a_u = R^{-1} \wt \eta_{\KL}(P_t, \mu)^{d(u,\rho)} r_u.
	\end{align*}
	For any $v\in c(u)$, by Markov chain
	\begin{align*}
		\sigma_{L_{v,h}} \to \sigma_v\to \sigma_u
	\end{align*}
	and discussion before Lemma~\ref{LemmaSmoothContCoef}, we have
	\begin{align*}
		I(\sigma_u; \sigma_{L_{v,h}}) \le \wt \eta_{\KL}(P_t, \mu) I(\sigma_{v},\sigma_{L_{v,h}}).
	\end{align*}
	Because $(\sigma_{L_{v,h}})_{v\in c(u)}$ are independent conditioned on $\sigma_u$ (where $c(u)$ is the set of children of $u$), we have
	\begin{align*}
		I(\sigma_u; \sigma_{L_{u,h}}) \le \sum_{v\in c(u)} I(\sigma_u; \sigma_{L_{v,h}}).
	\end{align*}
	Combining the two inequalities and let $h\to \infty$, we get
	\begin{align*}
		a_u \le \sum_{v\in c(u)} a_v.
	\end{align*}
	Furthermore, we have $a_u \le \wt \eta_{\KL}(P_t, \mu)^{d(u,\rho)}$.

	Now define a flow $b$ as following. For any $u\in V(T)$, let $u_0=\rho, \ldots, u_d=u$ be the shortest path from $\rho$ to $u$. Define
	\begin{align*}
		b_u = a_u \prod_{0\le j\le d-1} \frac{a_{u_j}}{\sum_{v\in c(u_j)}a_v}.
	\end{align*}
	(If for some $j$, we have $\sum_{v\in c(u_j)} a_v=0$, then let $b_u = 0$.)
	Then we have $$b_u = \sum_{v \in c(u)} b_v,$$ and that $$b_u \le a_u \le \wt \eta_{\KL}(P_t,\mu)^{d(u,\rho)}.$$
	By definition of branching number, we must have $b_\rho = 0$.
	Therefore $r_\rho = 0$ and non-reconstruction holds.

	\textbf{Galton-Watson tree case:}
	Let $P_d$ be the offspring distribution.
	We have
	\begin{align*}
		\bE_T I(\sigma_{\rho}; \sigma_{L_h}) &\le \bE_T \sum_{v\in c(\rho)} I(\sigma_{\rho}; \sigma_{L_{v,h}}) \\
		&\le \bE_T \sum_{v\in c(\rho)} \wt \eta_{\KL}(P_t,\mu) I(\sigma_v; \sigma_{L_{v,h}}) \\
		&= \wt \eta_{\KL}(P_t,\mu) \bE_{c(\rho)} [\sum_{v\in c(\rho)} \bE_{T_v} I(\sigma_v; \sigma_{L_{v,h}})] \\
		&= \wt \eta_{\KL}(P_t,\mu) \bE_{d\sim P_d} [d \bE_T I(\sigma_\rho; \sigma_{L_{h-1}})] \\
		&= \wt \eta_{\KL}(P_t,\mu) b \bE_T I(\sigma_\rho; \sigma_{L_{h-1}}).
	\end{align*}
	Here $T_v$ denotes the subtree rooted at $v$.
	Because $\bE_T I(\sigma_\rho; \sigma_{L_1}) < \infty$, when $b \wt \eta_{\KL}(P_t,\mu) < 1$, we have
	\begin{align*}
		\lim_{h\to \infty}\bE_T I(\sigma_{\rho}; \sigma_{L_h}) = 0.
	\end{align*}
	This finishes the proof.
\end{proof}

\section*{Acknowledgement}
The authors are grateful to Jingbo Liu for helpful discussions on reconstruction problems on trees, and to anonymous reviewers for valuable comments.
This work was supported in part by the MIT-IBM Watson AI Lab,
by the Center for Science of Information (CSoI), an National Science Foundation Science and Technology Center, under grant agreement CCF-09-39370,
and by the National Science Foundation under Grant No CCF-2131115.

\bibliographystyle{alpha}
\bibliography{ref}
\end{document}